\documentclass[journal]{IEEEtran}
\usepackage{amsthm}
\usepackage{empheq}
\usepackage{cancel}
\usepackage{mdframed,xcolor}
\usepackage{tcolorbox}
\usepackage{placeins}
\usepackage{relsize}
\usepackage{setspace}
\usepackage{stmaryrd}
\usepackage{xcolor}
\usepackage{url}
\usepackage{multirow}
\usepackage{amssymb}
\usepackage{amsmath}
\usepackage[end]{algpseudocode}
\usepackage{algorithm}
\usepackage{amsfonts}
\usepackage{diagbox}
\usepackage{booktabs}
\usepackage{filecontents}
\usepackage{mathtools}
\usepackage{caption}
\usepackage{subcaption}
\usepackage{enumitem}
\usepackage{xfrac}
\usepackage{nicefrac}
\usepackage{soul}
\usepackage{algorithm}
\algdef{SE}{Begin}{End}{\textbf{begin}}{\textbf{end}}
\setlist[itemize]{leftmargin=*}
\usepackage{xcolor}
\definecolor{rv1}{rgb}{1.0, 0.44, 0.37}
\definecolor{rv2}{rgb}{0.4, 1.0, 0.0}
\definecolor{rv3}{rgb}{0.0, 0.75, 1.0}
\definecolor{rvt}{rgb}{0.75, 0.75, 0.75}
\definecolor{my-blue}{cmyk}{0.1, 0.0, 0.0, 0.0, 1.00}
\usepackage{soul}

\newtheoremstyle{exampstyle}
  {7pt} 
  {7pt} 
  {\itshape} 
  {} 
  {\bfseries} 
  {.} 
  {.5em} 
  {} 
\theoremstyle{exampstyle}
\newtheorem{theorem}{Theorem}

\newtheorem{prepos}{Proposition}
\newtheorem{lemma}{Lemma}
\newtheorem{definition}{Definition}
\DeclareMathOperator*{\argmin}{arg\,min~}
\newcommand{\algrule}[1][.2pt]{\par\vskip.5\baselineskip\hrule height #1\par\vskip.5\baselineskip}

\newcommand{\Ab}{{\bf A}}

\newcommand{\Mb}{{\bf M}}

\newcommand{\xb}{{\bf x}}

\newcommand{\yb}{{\bf y}}
\newcommand{\zb}{{\bf z}}

\newcommand{\Db}{{\bf D}}
\newcommand{\Bb}{{\bf B}}

\newcommand{\bb}{{\bf b}}

\newcommand{\Xb}{{\bf X}}

\newcommand{\Wb}{{\bf W}}

\newcommand{\thetab}{{\mbox{\boldmath $\theta$}}}

\newcommand{\Ub}{{\mathbf U}}

\newcommand{\vb}{{\mathbf v}}
\newcommand{\Vb}{{\mathbf V}}

\newsavebox\mybox

\begin{document}

\title{Over-fitting Control in Deep Nonlinear Matrix Completion via Gradual Nonsmooth Regularization
}

\author{Sajad Faramarzi, Farzan Haddadi, Sajjad Amini, \\ Masoud Ahookhosh, and Symeon Chatzinotas, ~\IEEEmembership{Fellow,~IEEE}
\thanks{S. Faramarzi (e-mail: sajad.faramarzi1397@gmail.com) and F. Haddadi (e-mail: farzanhaddadi@iust.ac.ir) are with the School of Electrical Engineering, Iran University of Science $\&$ Technology, Tehran, Iran. S. Amini (e-mail: s\_amini@sharif.edu) is affiliated with the Electronics Research Institute (ERI) at Sharif University of Technology, Tehran, Iran. Additionally, S. Amini is a Visiting Faculty member at the College of Information and Computer Sciences, University of Massachusetts Amherst, Amherst, MA, USA. M. Ahookhosh (e-mail: masoud.ahookhosh@uantwerp.be) is with Department of Mathematics at the University of Antwerp, Antwerp, Belgium. S. Chatzinotas (e-mail: symeon.chatzinotas@uni.lu) is with the Interdisciplinary Centre for Security, Reliability and Trust, University of Luxembourg, Luxembourg City, Luxembourg.}}


\maketitle

\begin{abstract}
	Conventional matrix completion methods exploit the underlying low-rank structure of a target matrix to recover its missing entries, but may be inadequate for capturing complex nonlinear dependencies. Deep neural networks have therefore been increasingly adopted for nonlinear matrix completion (NLMC). However, their high representational capacity, particularly in deep fully connected neural networks (FCNNs), can lead to over-fitting when training is restricted to incomplete observations. In this paper, we propose a Deep Neural Network with Nonsmooth Regularization (DNN-NSR) framework to address this issue by incorporating two complementary structural priors: sparsity of intermediate representations and low-rankness of layer-wise weight matrices. Specifically, the $\ell_1$-norm of hidden representations promotes sparse latent features, while the nuclear norm of weight matrices encourages low-rank layer-wise transformations. To optimize the resulting nonsmooth and nonconvex objective, we develop an extrapolated proximal gradient algorithm that employs soft-thresholding and singular-value thresholding to handle the two regularization terms, respectively. We further introduce a gradual regularization strategy that progressively increases the influence of these structural penalties during training. Under suitable assumptions, we establish convergence properties of the proposed algorithm toward a critical point of the training objective. The effectiveness of DNN-NSR is evaluated through extensive experiments on synthetic matrix recovery, single-image inpainting, and recommender systems using the MovieLens 100K and MovieLens 1M datasets. Comparisons with representative linear and nonlinear matrix completion methods demonstrate the potential of the proposed framework to improve recovery performance by jointly exploiting sparse latent representations and low-rank network transformations.
\end{abstract}
\vspace*{-0.2em}


\vspace*{-1em}

\IEEEpeerreviewmaketitle

\vspace*{-0.5em}

\section{Introduction}
\vspace*{-0.5em}
\IEEEPARstart{M}{atrix} completion (MC) aims to recover the unobserved entries of a data matrix from a subset of its observed entries \cite{exact_matrix_completion, power_of_convex_relaxation, high_coherence_matrix_completion, non_convex_relaxation_adaptive, non_linear_matrix_completion, deep_learning_auto_encoder}. It has applications in a wide range of fields, including image inpainting \cite{image_inpainting_overview, light_image_inpainting}, wireless sensor network localization \cite{wireless_sensor_network, localization}, and recommender systems \cite{recommender_system_not_random, collaborative_filtering_factorization}. The central challenge is to infer missing entries by exploiting the underlying structure of the data. A major line of research addresses this problem through low-rank matrix completion (LRMC), which assumes that the target matrix admits a low-dimensional representation. Although this assumption has led to effective recovery methods, it may be restrictive when the data exhibit complex nonlinear dependencies. This limitation has motivated the development of nonlinear matrix completion (NLMC) approaches that employ nonlinear models to capture more complex relationships in the data.

\vspace*{-1.2em}
\subsection*{A. Linear Matrix Completion and Its Limitations}

Rank minimization and matrix factorization (MF) are two principal approaches to LRMC. Given an incomplete data matrix $\Xb \in \mathbb{R}^{m \times n}$, rank minimization seeks a matrix of minimum rank that agrees with the observed entries:
\small
\begin{equation}\label{rank_minimization}
\min_{\Mb} \quad \textrm{rank}(\Mb)
\quad
\textrm{s.t.} \quad
\Mb_{ij} = \Xb_{ij}, \quad \forall (i,j) \in \Omega,
\end{equation}
\normalsize
where $\Mb \in \mathbb{R}^{m \times n}$ denotes the matrix to be recovered, $\Omega$ is the set of observed entry indices, and $\Xb_{ij}$ and $\Mb_{ij}$ denote the $(i,j)$-th entries of the observed and recovered matrices, respectively. Since the rank function is nonconvex and discontinuous, solving \eqref{rank_minimization} is generally NP-hard \cite{exact_matrix_completion}. A widely adopted convex relaxation replaces the rank function with the nuclear norm, leading to
\small
\begin{equation}\label{nuclear_norm_minimization}
\min_{\Mb} \quad \|\Mb\|_*
\quad
\textrm{s.t.} \quad
\Mb_{ij} = \Xb_{ij}, \quad \forall (i,j) \in \Omega,
\end{equation}
\normalsize
where $\|\cdot\|_*$ denotes the nuclear norm, defined as the sum of the singular values. Singular value thresholding (SVT) \cite{svt1} and the inexact augmented Lagrange multiplier (IALM) method \cite{IALM} are representative approaches for solving this formulation. Other methods seek closer approximations to the rank function. For example, the truncated nuclear norm (TNN) \cite{TNN} and its weighted-residual variant, TNNWRE \cite{TNNWRE}, modify the treatment of singular values to reduce the bias associated with conventional nuclear norm minimization.

More generally, the rank function and nuclear norm correspond to the Schatten-$p$ quasi-norm for $p=0$ and $p=1$, respectively. For $0<p<1$, the Schatten-$p$ quasi-norm provides a nonconvex surrogate that can approximate the rank more closely than the nuclear norm, although it leads to a more challenging optimization problem \cite{Schatten}. Despite their effectiveness, nuclear norm- and Schatten-$p$-based methods typically require singular value decomposition (SVD) computations during optimization, which can be computationally demanding for large-scale matrices \cite{complexity_of_nuclear_norm1, complexity_of_nuclear_norm2}. Several alternatives have therefore been developed to alleviate this burden. The factored nuclear norm \cite{factor_group1, factor_group2}, for instance, exploits a factorized representation to avoid repeated SVD computations. Adaptive correlation learning (NCARL) \cite{NCARL} employs a nonconvex surrogate of the rank with closed-form intermediate updates, whereas capped reweighting norm minimization (CRNM) \cite{CRNM} adaptively reweights singular values and establishes convergence guarantees.

MF provides another widely used formulation of LRMC. Specifically, a low-rank matrix $\Mb$ can be represented as the product of two lower-dimensional factor matrices, $\Mb=\Ab\Bb$, where $\Ab\in\mathbb{R}^{m\times r}$, $\Bb\in\mathbb{R}^{r\times n}$, and $r<\min(m,n)$. A conventional MF formulation for MC is given by \cite{MF1}
\small
\begin{equation}\label{MF1}
\min_{\Ab,\Bb} \quad
\frac{1}{2}\left\|
\mathcal{P}_{\Omega}(\Ab\Bb-\Xb)
\right\|_{F}^{2},
\end{equation}
\normalsize
where $\mathcal{P}_{\Omega}(\cdot)$ denotes the projection onto the observed entries. Both rank-based formulations and conventional MF exploit low-dimensional linear representations of the data \cite{linear_MC}. Although effective for matrices governed by low-rank structures, these approaches may be insufficient when the underlying data exhibit complex nonlinear dependencies \cite{NLMC_image_inpainting, NLMC_Recommender_system}. Recent extensions, such as asymmetric projected gradient descent (APGD) \cite{APGA} and multi-matrix completion (MMC) \cite{MMC}, improve recovery or exploit correlations across matrices. Nevertheless, the need to capture nonlinear dependencies motivates the development of NLMC approaches.
\vspace*{-1.1em}
\subsection*{B. Nonlinear Matrix Completion and Deep Estimators}

NLMC methods aim to capture dependencies that cannot be adequately represented by linear low-rank models. Kernel-based approaches, for example, transform the original data into a feature space in which nonlinear relationships can be modeled through linear operations. Kernelized probabilistic matrix factorization (KPMF) \cite{KPMC} represents the incomplete matrix through latent factors whose covariance structures are defined by kernel functions derived from side information. Related approaches exploit Schatten-$p$ regularization of the resulting low-rank representation to recover missing entries \cite{NLMC}. Although kernel methods provide a principled means of incorporating nonlinear relationships, their performance depends on the choice of kernel functions and the availability of suitable side information.

Artificial neural networks offer a flexible alternative because they can learn nonlinear mappings directly from observed data. Different neural architectures, including restricted Boltzmann machines, autoencoders (AEs), and recurrent neural networks, have consequently been investigated for NLMC \cite{RBM, AECF, recurrent_NLMC}. Among these architectures, feed-forward neural networks, particularly multilayer perceptrons and fully connected neural networks (FCNNs), have attracted considerable attention because of their relatively simple structures and ability to represent nonlinear relationships \cite{AECF, AEMC_DLMC, CSRMC, MCKR, bi_branch}.

The autoencoder-based collaborative filtering (AECF) method \cite{AECF} is an early neural approach that replaces missing entries with a predefined constant before processing the data through an AE. This treatment may introduce bias because the network also processes artificially assigned values. Autoencoder-based matrix completion (AEMC) \cite{AEMC_DLMC} addresses this issue by constructing the training loss using only the outputs associated with observed entries. Its deeper extension, DLMC \cite{AEMC_DLMC}, further increases the representational capacity of the network. Bi-branch neural networks (BiBNNs) \cite{bi_branch} adopt an MF-based architecture in which the $i$-th row of $\Ab$ and the $j$-th column of $\Bb$ are processed by two separate neural branches to produce a nonlinear estimate $\widehat{\Mb}_{ij}$. A related approach \cite{mehrdad} combines a nonlinear transformation $f(\Ab)$ with the bilinear term $\Ab\Bb$ in a regularized two-branch architecture. More recently, learned robust matrix completion (LeRMC) \cite{LeRMC} has incorporated deep unfolding to learn algorithmic parameters and reduce computational costs, reporting competitive performance in applications such as video background subtraction, face modeling, and cloud removal.

Despite their flexibility in modeling nonlinear dependencies, deep neural estimators are susceptible to over-fitting, particularly when trained using incomplete observations. This challenge motivates the development of effective regularization strategies for deep NLMC.
\vspace*{-1em}
\subsection*{C. Over-fitting in Deep NLMC and Structure-Aware Regularization}

Over-fitting is a major concern in deep NLMC because training is restricted to the observed entries, leaving the missing entries without direct supervision. A highly expressive network may therefore fit the available observations without learning relationships that generalize to unobserved entries. This risk becomes more pronounced as the observation rate decreases, making structural regularization essential for controlling model complexity and improving recovery performance.

Several regularization strategies have been investigated in neural matrix completion. In \cite{AEMC_DLMC}, the Frobenius norm of the network weight matrices is used as a smooth penalty. In \cite{CSRMC}, two AEs process row and column vectors, respectively, and an additional penalty based on the cosine similarity between their code-layer representations is incorporated into training. Matrix completion kernel regularization (MCKR) \cite{MCKR} instead maps the nonlinear problem into a high-dimensional feature space, where a linear model can be employed. These methods demonstrate different ways of constraining neural estimators. However, they do not jointly impose the two structural priors considered in this work: sparsity of intermediate latent representations and low-rankness of layer-wise weight matrices.

These two properties provide complementary forms of regularization for a deep NLMC model. First, sparsity in intermediate representations restricts the latent features activated by the network and can be promoted through the $\ell_1$-norm. Second, low-rankness of the weight matrices constrains the dimensionality of the layer-wise transformations and can be encouraged through the nuclear norm. Importantly, these regularizers act on different components of the neural estimator: the $\ell_1$-norm is applied to hidden representations, whereas the nuclear norm is imposed on network weights. Their joint use therefore constrains both the learned representations and the transformations that generate them.

The use of nonsmooth regularization in neural-network training has been investigated in earlier studies. In particular, \cite{amini2019} studied the pretraining of autoencoders with a single nonsmooth regularizer in a two-layer FCNN, while \cite{amini2022} extended this line of work to feed-forward networks with one nonsmooth regularizer and employed extrapolation to accelerate the optimization process. These studies provide important foundations for nonsmooth neural-network training. However, the problem considered here is specifically deep nonlinear matrix completion, where incomplete observations motivate the joint incorporation of sparsity in hidden representations and low-rankness in network transformations. This joint formulation leads to a nonsmooth and nonconvex training objective, requiring an optimization method that can accommodate both regularization terms and provide convergence guarantees.
\vspace*{-1em}
\subsection*{D. Contributions of This Work}

In this paper, we propose a Deep Neural Network with Nonsmooth Regularization (DNN-NSR) model for nonlinear matrix completion. The proposed model employs an FCNN as a nonlinear estimator and incorporates two complementary structural priors: sparsity of intermediate representations and low-rankness of layer-wise weight matrices. To optimize the resulting nonsmooth and nonconvex objective, we develop an extrapolated proximal gradient algorithm and introduce a gradual regularization strategy. We also establish convergence properties of the proposed algorithm. The main contributions are summarized as follows.

\begin{itemize}
	
	\item \textbf{A structure-aware regularization formulation for deep nonlinear matrix completion.}
	We formulate a deep NLMC model that jointly incorporates sparsity and low-rankness as complementary structural priors. Specifically, the $\ell_1$-norm of the intermediate representations across all hidden layers promotes sparse latent features, whereas the nuclear norm of the layer-wise weight matrices encourages low-rank transformations within the FCNN. The two terms constrain different components of the nonlinear estimator and are jointly used to control its effective representational complexity. In this way, the proposed formulation goes beyond generic weight regularization by explicitly incorporating structural constraints on both the learned representations and the transformations that generate them. To the best of our knowledge, this is the first work to investigate the joint use of these two nonsmooth regularizers within a deep FCNN-based matrix completion framework.
	
	\item \textbf{A gradual regularization strategy for training deep NLMC models.}
	We introduce a gradual learning strategy that progressively incorporates the nonsmooth regularization terms into network training. The network is initially allowed to fit the observed entries with limited regularization, after which the influence of the structural penalties is increased. This strategy is designed to balance data fitting and structural regularization, allowing the nonlinear estimator to learn the dominant patterns in the observations before stronger sparsity and low-rankness constraints are imposed. We investigate the effect of this training strategy experimentally and demonstrate its contribution to the recovery performance of the proposed method.
	
	\item \textbf{An extrapolated proximal gradient algorithm for the proposed nonsmooth objective.}
	The simultaneous use of the $\ell_1$- and nuclear norm penalties results in a nonsmooth, nonconvex optimization problem. We develop an extrapolated proximal gradient algorithm that accommodates the two regularization terms through their respective proximal operators: soft-thresholding for the hidden representations and singular-value thresholding for the weight matrices. In addition, we derive an adaptive rule for the extrapolation weight based on the descent conditions established in the convergence analysis. This rule avoids the need to prescribe a fixed extrapolation weight manually and provides a principled means of controlling the extrapolation step during optimization.
	
	\item \textbf{Convergence analysis of the proposed DNN-NSR model.}
	We establish the convergence properties of the proposed extrapolated proximal gradient algorithm for the DNN-NSR model. Under suitable assumptions, we prove that the sequence generated by the algorithm converges to a critical point of the resulting nonsmooth and nonconvex training objective. In particular, we establish the required descent and regularity properties of the objective function and verify its Kurdyka--\L{}ojasiewicz (KL) property.
	
	\item \textbf{Comprehensive numerical evaluation on matrix completion tasks.}
	We evaluate the proposed DNN-NSR method across three matrix completion settings: synthetic matrix recovery, single-image inpainting, and recommender systems using the MovieLens 100K and MovieLens 1M datasets. Its performance is compared with six representative matrix completion algorithms, including two classical LRMC methods, IALM and NCARL, and four NLMC methods based on deep neural networks, namely AEMC, DLMC, BiBNN, and LeRMC. The experiments investigate the recovery performance of these methods under different observation rates and across different data types, providing an assessment of the proposed structure-aware regularization framework in both synthetic and practical matrix completion scenarios. The results demonstrate the effectiveness of the proposed DNN-NSR method in exploiting sparse latent representations and low-rank weight structures through gradual nonsmooth regularization.
	
\end{itemize}

The remainder of this paper is organized as follows. Section II presents the notation and mathematical preliminaries. Section III reviews prior work and motivates the proposed structural regularization. Section IV introduces the DNN-NSR model and algorithm and establishes its convergence properties. Section V presents the experimental results. Finally, Section VI concludes the paper.

\section{Notations and Preliminaries}
\label{sec:preliminaries}

This section introduces the mathematical notation and preliminary results required to formulate the proposed DNN-NSR model and analyze the convergence properties of its optimization algorithm. We first establish the notation used throughout the paper and review the proximal operators associated with the nonsmooth regularization terms and parameter constraints. We then recall the smoothness and descent properties underlying the proximal updates, followed by several auxiliary inequalities used in the convergence analysis. Finally, we introduce the notions of critical points and the KL property, which provide the theoretical framework for establishing convergence of the entire sequence generated by the proposed algorithm.
\vspace*{-1em}
\subsection*{A. Notation}

Throughout the paper, bold uppercase and bold lowercase letters represent matrices and vectors, respectively. The $i$-th component of a vector $\xb$ is denoted by $x_i$, whereas $w_{ij}$ denotes the entry of a matrix $\Wb$ in its $i$-th row and $j$-th column. The operator $\mathrm{vec}(\cdot)$ stacks the columns of a matrix into a single column vector, and $\mathrm{vec}^{-1}(\cdot)$ denotes its inverse. The column-wise and row-wise concatenations of $m$ vectors are respectively expressed as $\xb = [\xb_1, \dots, \xb_m]$ and $\yb = [\yb_1; \dots; \yb_m]$. The symbols $\lVert \cdot \rVert_p$ and $\lVert \cdot \rVert_F$ denote the $\ell_p$ norm and Frobenius norm, respectively. In particular, the $\ell_1$ norm $\lVert \cdot \rVert_1$ is defined as the sum of the absolute values of all entries, whereas the nuclear norm $\lVert \cdot \rVert_*$ is defined as the sum of the singular values of a matrix. The Hadamard product is denoted by $\odot$, and $\mathrm{dom}(f)$ represents the domain of a function $f$.
\vspace*{-1em}
\subsection*{B. Proximal Operators}

Proximal operators constitute the main computational tools for handling the nonsmooth terms in the proposed optimization framework. In particular, the proximal updates developed in Section IV rely on closed-form expressions for the $\ell_1$ norm, the nuclear norm, and the indicator function of an $\ell_\infty$ ball. We first recall the general definition of a proximal operator.

\begin{definition}[Proximal Operator \cite{proximal_operator}]
	\label{def:prox}
	Let $g: \mathbb{R}^{n} \rightarrow \mathbb{R} \cup \{+\infty\}$ be a proper and lower semi-continuous function. For $\alpha, \mu > 0$, the proximal operator associated with $g(\xb) = \alpha \mu f(\xb)$ is defined as
	\small
	\begin{equation}
	\mathrm{Prox}_{g}(\xb) = \argmin_{\yb \in \mathbb{R}^n} \left( \alpha f(\yb) + \frac{1}{2\mu} \lVert \yb - \xb \rVert^{2} \right).
	\end{equation}
	\normalsize
\end{definition}

The proximal operator generalizes the projection onto a convex set and admits closed-form solutions for several functions relevant to the proposed method. The following results summarize the proximal mappings associated with the nonsmooth regularizers and the parameter constraint considered in this work.

\small
\begin{lemma}[Proximal Operator of the $\ell_1$ Norm \cite{amini2018}]
	\label{lemma:l1prox}
	The proximal operator of $g(\xb) = \alpha \lVert \xb \rVert_1$ is given by the soft-thresholding operator, whose $i$-th component is
	\begin{equation}
	\label{eq:softthreshold}
	\big[\mathrm{Prox}_{g}(\xb)\big]_{i} =
	\begin{cases}
	x_i - \alpha, & x_i > \alpha, \\
	0, & |x_i| \le \alpha, \\
	x_i + \alpha, & x_i < -\alpha.
	\end{cases}
	\end{equation}
\end{lemma}

\normalsize
The soft-thresholding operation in \eqref{eq:softthreshold} provides the mechanism through which the $\ell_1$ regularizer promotes sparsity in the hidden-layer representations. Specifically, components with magnitudes below the threshold $\alpha$ are set to zero, whereas the remaining components are shrunk toward the origin.

\begin{definition}[Singular Value Shrinkage Operator \cite{svt1}]
	\label{def:svshrink}
	Consider the singular value decomposition (SVD) of a matrix $\Xb \in \mathbb{R}^{m \times n}$, given by
	\begin{equation}
	\label{eq:svd}
	\Xb = \Ub \mathbf{\Sigma} \Vb^{T}, \quad
	\mathbf{\Sigma} = \mathrm{diag}\big(\{\sigma_{i}\}_{1 \le i \le \min(m,n)}\big).
	\end{equation}
	The singular value shrinkage operator $\Db_{\tau}$ is defined as
	\begin{equation}
	\label{eq:svshrink}
	\Db_{\tau}(\Xb) = \Ub \mathbf{\Sigma}_{\tau} \Vb^{T}, \quad
	\mathbf{\Sigma}_{\tau} = \mathrm{diag}\big(\max\{\sigma_{i}-\tau, 0\}\big).
	\end{equation}
\end{definition}

\begin{lemma}[Proximal Operator of the Nuclear Norm \cite{svt1}]
	\label{lemma:nucprox}
	For $g(\Xb) = \beta \lVert \Xb \rVert_{*}$, where $\beta > 0$, the corresponding proximal operator is given by
	\begin{equation}
	\mathrm{Prox}_{g}(\Xb) = \Db_{\beta}(\Xb).
	\end{equation}
\end{lemma}

Lemma~\ref{lemma:nucprox} establishes that the proximal mapping of the nuclear norm reduces to applying soft-thresholding to the singular values of its matrix argument. In the proposed framework, this operation is used to promote low-rank structures in the network weight matrices. Specifically, singular values below the threshold $\beta$ are set to zero, while the remaining singular values are reduced by $\beta$, thereby encouraging a lower effective rank in the corresponding layer-wise transformations.

\begin{lemma}[Proximal Operator of the $\ell_\infty$ Indicator \cite{amini2019}]
	\label{lemma:infprox}
	Let $r_{\theta}(\thetab)$ denote the indicator function of the $\ell_\infty$ ball with radius $M > 0$, defined as
	\small
	\begin{equation}
	\label{eq:indicator}
	r_{\theta}(\thetab) =
	\begin{cases}
	0, & \lVert \thetab \rVert_{\infty} \le M, \\
	+\infty, & \lVert \thetab \rVert_{\infty} > M,
	\end{cases}
	\end{equation}
	\normalsize
	where $\lVert \thetab \rVert_\infty = \max_i |\theta_i|$. For $\mu_{\theta,k} > 0$, the proximal operator $\mathbf{u} = \mathrm{Prox}_{\mu_{\theta,k}, r_{\theta}}(\thetab)$ is given component-wise by
	\small
	\begin{equation}
	\label{eq:projinf}
	u_i =
	\begin{cases}
	\theta_i, & \theta_i \le M, \\
	M \, \mathrm{sign}(\theta_i), & \theta_i > M.
	\end{cases}
	\end{equation}
	\normalsize
\end{lemma}

The indicator function $r_{\theta}(\thetab)$ plays an important role in the convergence analysis by restricting the network parameters to a bounded set. This boundedness is subsequently used to establish the Lipschitz continuity of the gradient of the smooth part of the objective function. Accordingly, the proximal mapping in \eqref{eq:projinf} corresponds to projecting the parameter vector $\thetab$ onto the $\ell_\infty$ ball of radius $M$.

\subsection*{C. Smoothness and Descent}

The convergence analysis in Section IV relies on the Lipschitz continuity of the gradient of the smooth part of the objective function. This property provides the quadratic upper bound required to derive the proximal updates and establish their descent behavior. We first recall the definition of a smooth function and then establish the smoothness property on compact sets.

\begin{definition}[$L$-Smooth Functions \cite{lipschitz_gradient}]
	\label{def:lsmooth}
	A function $f: \mathbb{R}^{n} \rightarrow \mathbb{R}$ is called $L$-smooth if there exists a constant $L > 0$ such that
	\small
	\begin{equation}
	\lVert \nabla f(\xb) - \nabla f(\yb) \rVert \le L \lVert \xb - \yb \rVert, \quad \forall \xb, \yb \in \mathbb{R}^n.
	\end{equation}
	\normalsize
\end{definition}

\begin{lemma}[Smoothness on Compact Sets]
	\label{lem:compact}
	Let $f: \mathbb{R}^{n} \rightarrow \mathbb{R}$ be twice continuously differentiable ($f \in \mathcal{C}^2$), and let $\mathcal{C}$ be a convex compact set. Then, $f$ is $L$-smooth on $\mathcal{C}$ for some $L > 0$.
\end{lemma}

\begin{proof}
	Since $f \in \mathcal{C}^2$, its Hessian $\nabla^2 f$ is continuous. Moreover, the compactness of $\mathcal{C}$ ensures the existence of a constant $L > 0$ such that $\lVert \nabla^2 f(\mathbf{z}) \rVert \le L$ for every $\mathbf{z} \in \mathcal{C}$. For any $\mathbf{x}, \mathbf{y} \in \mathcal{C}$ and $\tau \in [0,1]$, the convexity of $\mathcal{C}$ implies that $\tau \mathbf{x} + (1-\tau)\mathbf{y} \in \mathcal{C}$. Applying the mean value theorem along the line segment connecting $\mathbf{x}$ and $\mathbf{y}$ yields
	\small
	\begin{equation}
	\begin{split}
	\lVert \nabla f(\mathbf{x}) - \nabla f(\mathbf{y}) \rVert
	&\le \int_{0}^{1} \lVert \nabla^{2} f(\tau \mathbf{x} + (1-\tau)\mathbf{y}) \rVert \, \lVert \mathbf{x}-\mathbf{y} \rVert \, d\tau \\
	&\le L \lVert \mathbf{x}-\mathbf{y} \rVert,
	\end{split}
	\end{equation}
	\normalsize
	which establishes the desired result.
\end{proof}

Lemma~\ref{lem:compact} is particularly relevant to the proposed framework. Since the indicator function $r_{\theta}(\thetab)$ restricts the network parameters to a bounded set, the smooth part of the objective has a Lipschitz-continuous gradient on this set, provided that it is twice continuously differentiable. This property supports the selection of admissible step sizes in the proximal updates.

\begin{lemma}[Descent Lemma \cite{Descent_Lemma_A}]
	\label{lemma:descent}
	Let $f: \mathbb{R}^{n} \rightarrow \mathbb{R}$ be continuously differentiable on a neighborhood of every point in a convex set $\mathcal{C} \subseteq \mathbb{R}^n$ with nonempty interior, and suppose that $f$ is $L$-smooth on $\mathcal{C}$ for some $L > 0$. Then, for any $\mathbf{x}, \mathbf{y} \in \mathcal{C}$,
	\small
	\begin{equation}
	f(\mathbf{y}) \le f(\mathbf{x}) + \nabla^{T} f(\mathbf{x}) (\mathbf{y} - \mathbf{x}) + \frac{L}{2} \lVert \mathbf{y} - \mathbf{x} \rVert^{2}.
	\end{equation}
	\normalsize
\end{lemma}

The descent lemma provides the quadratic upper bound used to construct the proximal updates in Section IV. In particular, it enables the smooth part of the objective to be approximated by a first-order expansion with a quadratic error term, thereby facilitating the derivation of the iterative optimization scheme.

\subsection*{D. Auxiliary Inequalities}

The following two elementary inequalities are repeatedly used in the convergence analysis, particularly when deriving bounds for the terms arising from the proximal and extrapolation steps.

\begin{lemma}[Three-Point Identity]
	\label{lemma:threepoint}
	For any $\mathbf{x}, \mathbf{y}, \mathbf{z} \in \mathbb{R}^n$,
	\begin{equation}
	\label{eq:threepoint}
	\lVert \mathbf{x}-\mathbf{z} \rVert^{2} - \lVert \mathbf{y}-\mathbf{z} \rVert^{2}
	= 2(\mathbf{x} - \mathbf{y})^{T} (\mathbf{y} - \mathbf{z}) + \lVert \mathbf{x} - \mathbf{y} \rVert^{2}.
	\end{equation}
\end{lemma}

\begin{lemma}[Young's Inequality \cite{young_inequality}]
	\label{lemma:young}
	For any $\epsilon > 0$ and any $a, b \in \mathbb{R}$,
	\small
	\begin{equation}
	ab \le \frac{a^{2}}{2\epsilon} + \frac{\epsilon b^{2}}{2}.
	\end{equation}
	\normalsize
\end{lemma}

The three-point identity in Lemma~\ref{lemma:threepoint} is used to decompose differences of squared distances arising in the proximal analysis. Meanwhile, Young's inequality in Lemma~\ref{lemma:young} is employed to bound cross terms introduced by the extrapolation step. Together, these inequalities facilitate the derivation of the descent estimates required for the convergence analysis.

\subsection*{E. Optimality and the Kurdyka--\L{}ojasiewicz Property}

Since the objective function of the proposed DNN-NSR model is nonsmooth and nonconvex, convergence is characterized in terms of critical points rather than global minimizers. We first recall the definition of a critical point based on the limiting subdifferential and then introduce the KL property, which is used to establish convergence of the entire sequence generated by the algorithm.

\begin{definition}[Critical Point \cite{critical_point}]
	\label{def:critical}
	A point $\xb^{*}$ is called a critical point of $f$ if $\mathbf{0} \in \partial f(\xb^{*})$, where $\partial f(\xb^{*})$ denotes the limiting subdifferential of $f$ at $\xb^{*}$.
\end{definition}

The KL property provides a quantitative characterization of the behavior of a function around its critical points. In the convergence analysis, it is used to strengthen the convergence of subsequences obtained from the descent analysis and establish convergence of the entire generated sequence.

\begin{definition}[Kurdyka--\L{}ojasiewicz Property \cite{KL_property}]
	\label{def:kl}
	A function $f(\xb)$ is said to satisfy the KL property at $\bar{\xb} \in \mathrm{dom}(\partial f)$ if there exist $\eta > 0$, a neighborhood $\mathcal{B}_{\rho}(\bar{\xb}) \triangleq \{\xb : \lVert \xb - \bar{\xb} \rVert < \rho\}$, and a concave function $g(a) = c a^{1-\theta}$, for some $c > 0$ and $\theta \in [0,1)$, such that
	\begin{equation}
	\label{eq:kl}
	g\big(|f(\xb)-f(\bar{\xb})|\big) \, \mathrm{dist}\big(\mathbf{0}, \partial f(\xb)\big) \ge 1
	\end{equation}
	for all $\xb \in \mathcal{B}_{\rho}(\bar{\xb}) \cap \mathrm{dom}(\partial f)$ satisfying $f(\bar{\xb}) < f(\xb) < f(\bar{\xb}) + \eta$, where
	\begin{equation}
	\mathrm{dist}\big(\mathbf{0}, \partial f(\xb)\big) = \min\big\{\lVert \vb \rVert_{2} : \vb \in \partial f(\xb)\big\}.
	\end{equation}
\end{definition}

The KL property is satisfied by a broad class of functions commonly encountered in signal processing and machine learning, including the $\ell_1$ norm, the nuclear norm, and indicator functions of suitable constraint sets. In Section IV, we establish that the objective function of the proposed model satisfies this property. This result enables us to invoke the global convergence theorem in \cite{whole_sequence} and establish convergence of the entire sequence generated by the proposed algorithm to a critical point.

\subsection*{F. Summary}

This section has introduced the notation and mathematical tools required for the formulation and convergence analysis of the proposed DNN-NSR algorithm. The proximal mappings summarized in Lemmas~\ref{lemma:l1prox}--\ref{lemma:infprox} provide the closed-form updates associated with the nonsmooth regularization terms and the parameter constraint. The smoothness and descent results in Lemmas~\ref{lem:compact} and \ref{lemma:descent} establish the conditions required to construct the quadratic upper bounds used in the proximal updates and determine admissible step sizes. Furthermore, the three-point identity and Young's inequality in Lemmas~\ref{lemma:threepoint} and \ref{lemma:young} are used to control the terms introduced by the extrapolation step. Finally, the critical-point and KL definitions in Definitions~\ref{def:critical} and \ref{def:kl} provide the theoretical basis for establishing convergence of the entire sequence to a critical point. These results collectively establish the mathematical preliminaries for the model, optimization algorithm, and convergence analysis presented in Section IV.

\section{Prior Art and Motivation}
\label{sec:prior}

\subsection*{A. Matrix Completion via Deep AEs}

Deep autoencoders (AEs), typically implemented as fully connected neural networks (FCNNs) with $p-1$ hidden layers, provide a nonlinear framework for matrix completion \cite{AEMC_DLMC}. Let $\Xb = [\xb_1, \dots, \xb_n] \in \mathbb{R}^{m \times n}$ denote a partially observed matrix, where each column $\xb_i \in \mathbb{R}^m$ is treated as an input to the network. Accordingly, both the input and output layers contain $m$ neurons. Given an input column $\xb_q$, the network produces an estimate $\widehat{\xb}_q$ through the mapping
\small
\begin{equation}
\label{eq:ae_output}
\widehat{\xb}_q = g^{(p)} \Biggl( \mathbf{W}^{(p)} \biggl( \dots \Bigl( g^{(1)} \bigl( \mathbf{W}^{(1)} \xb_q + \mathbf{b}^{(1)} \bigr) \Bigr) \dots \biggr) + \mathbf{b}^{(p)} \Biggr),
\end{equation}
\normalsize
where $g^{(i)}$, $\bb^{(i)}$, and $\Wb^{(i)}$ denote the activation function, bias vector, and weight matrix of the $i$-th layer, respectively \cite{AEMC_DLMC}. Equivalently, the network output can be expressed as a nonlinear function of its trainable parameters:
\small
\begin{equation}
\widehat{\xb}_q = F_1 \big( \Wb^{(1)}, \dots, \Wb^{(p)}, \mathbf{b}^{(1)}, \dots, \mathbf{b}^{(p)}; \xb_q \big).
\end{equation}
\normalsize
For convenience, all trainable parameters are collected into a single vector:
\small
\begin{equation}
\label{eq:theta}
\thetab = \big[ \mathrm{vec}(\Wb^{(1)}), \dots, \mathrm{vec}(\Wb^{(p)}), \mathbf{b}^{(1)}, \dots, \mathbf{b}^{(p)} \big].
\end{equation}
\normalsize
A natural training criterion is the mean squared error (MSE),
\small
\begin{equation}
\label{eq:mse}
\ell(\thetab) = \frac{1}{n} \sum_{q=1}^{n} \big\lVert \xb_q - \widehat{\xb}_q \big\rVert^2.
\end{equation}
\normalsize
However, this loss cannot be directly evaluated because only a subset of the entries of $\Xb$ is available. To account for the observation pattern, we define a binary mask matrix $\mathbf{N} \in \{0,1\}^{m \times n}$ as
\small
\begin{equation}
\label{eq:mask}
n_{ij} = \begin{cases} 1, & (i,j) \in \Omega, \\ 0, & (i,j) \notin \Omega, \end{cases}
\end{equation}
\normalsize
where $\Omega$ denotes the set of observed entries. Let $\mathbf{n}_q$ denote the $q$-th column of $\mathbf{N}$. The resulting masked loss is
\small
\begin{equation}
\label{eq:masked_loss}
\ell(\thetab) = \frac{1}{n} \sum_{q=1}^{n} \big\lVert \mathbf{n}_q \odot (\xb_q - \widehat{\xb}_q) \big\rVert^2,
\end{equation}
\normalsize
where $\odot$ denotes the Hadamard product. To mitigate over-fitting, a conventional approach is to introduce an $\ell_2$ regularization term on the network weights, yielding
\small
\begin{equation}
\label{eq:reg_loss}
\ell(\thetab) = \frac{1}{n} \sum_{q=1}^{n} \big\lVert \mathbf{n}_q \odot (\xb_q - \widehat{\xb}_q) \big\rVert^2 + \lambda \sum_{i=1}^{p} \big\lVert \Wb^{(i)} \big\rVert_F^2,
\end{equation}
\normalsize
or, equivalently, in matrix form,
\small
\begin{equation}
\label{eq:reg_loss_matrix}
\mathcal{L}(\thetab) = \big\lVert \mathbf{N} \odot (\Xb - \widehat{\Xb}) \big\rVert_F^2 + \lambda \sum_{i=1}^{p} \big\lVert \Wb^{(i)} \big\rVert_F^2,
\end{equation}
\normalsize
where $\widehat{\Xb} = [\widehat{\xb}_1, \dots, \widehat{\xb}_n]$. The network parameters are obtained by minimizing $\mathcal{L}(\thetab)$, after which the missing entries are estimated from the corresponding network outputs:
\begin{equation}
\label{eq:recovery}
x_{ij} = \hat{x}_{ij}, \quad (i,j) \notin \Omega.
\end{equation}

\subsection*{B. Structural Priors Beyond the Frobenius Norm}

Although the Frobenius-norm regularizer in \eqref{eq:reg_loss_matrix} provides a smooth penalty that can be efficiently handled by gradient-based methods, it does not explicitly promote sparsity in hidden representations or low-rank structure in layer-wise weight matrices. These two structural properties motivate the use of alternative regularization terms. In particular, $\ell_0$ regularization of hidden-layer activations has been shown to reduce the generalization error of autoencoders \cite{amini2018}. Similarly, low-rank weight matrices have been reported to accelerate learning \cite{low_rank1}, improve prediction accuracy \cite{low_rank2}, and control the effective capacity of deep networks \cite{low_rank3}. These findings motivate the incorporation of $\ell_1$ regularization on hidden representations and nuclear norm regularization on weight matrices, as developed in Section IV.

Beyond the choice of regularizers, their incorporation during training can also affect the resulting solution. Existing approaches commonly employ regularization terms with fixed strengths throughout training. However, as demonstrated in Section V-D, applying strong regularization too early may hinder the network from adequately fitting the observed entries, whereas introducing it too late may allow excessive fitting to the observed data. This motivates the gradual regularization strategy developed in Section IV, which progressively increases the influence of the regularization terms during training.

\section{Proposed Model and Method}
\label{sec:proposed}

Building on the limitations of conventional weight regularization discussed in Section II, we develop DNN-NSR by incorporating sparse latent representations and low-rank layer-wise transformations into the FCNN-based matrix completion framework. To efficiently handle the resulting nonsmooth terms, we introduce auxiliary variables and formulate the problem as a smooth nonconvex component combined with separable nonsmooth regularizers. We then develop a block-coordinate proximal scheme with closed-form updates for the auxiliary variables and a projected proximal-gradient update for the network parameters.

\subsection*{A. Nonsmooth Regularized Objective}

As discussed in Section II, conventional Frobenius-norm regularization does not explicitly promote sparsity in hidden representations or low-rank structures in network weight matrices. To impose these two complementary structural priors, we consider an FCNN with $l$ hidden layers and introduce an $\ell_1$ penalty on the hidden-layer outputs and a nuclear norm penalty on the weight matrices.

Ideally, sparsity in the hidden representations and low-rankness of the weight matrices could be promoted through the $\ell_0$ pseudo-norm and the rank function, respectively. However, directly optimizing these quantities is generally computationally intractable. We therefore employ their standard convex surrogates, namely, the $\ell_1$ norm and the nuclear norm. The resulting objective is
\small
\begin{equation}
\label{eq:proposed_obj}
\min_{\thetab} \ \mathcal{L}(\thetab)
+ \sum_{q=1}^{n} \sum_{i=1}^{l}
\alpha_i \big\lVert \mathbf{z}_q^{(i)} \big\rVert_0
+ \sum_{j=1}^{l+1}
\beta_j \, \mathrm{rank}\big(\Wb^{(j)}\big),
\end{equation}
\normalsize
where $\alpha_i > 0$ and $\beta_j > 0$ are regularization parameters, and
\small
\begin{equation}
\label{eq:zq}
\zb_q^{(i)} =
F_2 \big(
\Wb^{(1)}, \mathbf{b}^{(1)}, \dots,
\Wb^{(i)}, \mathbf{b}^{(i)}; \xb_q
\big)
\end{equation}
\normalsize
denotes the output of the $i$-th hidden layer for the $q$-th training sample. More explicitly,
\small
\begin{equation}
\label{eq:zq_expanded}
\zb_q^{(i)} = g^{(i)} \Biggl(
\mathbf{W}^{(i)}
\biggl( \dots
\Bigl(
g^{(1)} \bigl(
\mathbf{W}^{(1)} \xb_q + \mathbf{b}^{(1)}
\bigr)
\Bigr)
\dots \biggr)
+ \mathbf{b}^{(i)}
\Biggr).
\end{equation}
\normalsize
Replacing the $\ell_0$ pseudo-norm and the rank function with their convex surrogates yields the proposed nonsmooth regularized objective:
\small
\begin{equation}
\label{eq:proposed_P}
\begin{split}
\mathcal{P}: \quad
\min_{\thetab} \quad
& \mathcal{L}(\thetab)
+ \sum_{q=1}^{n} \sum_{i=1}^{l}
\alpha_i \big\lVert \zb_q^{(i)} \big\rVert_1 
+ \sum_{j=1}^{l+1}
\beta_j \big\lVert \Wb^{(j)} \big\rVert_* .
\end{split}
\end{equation}
\normalsize
Here, the $\ell_1$ terms promote sparse hidden-layer representations, whereas the nuclear norm terms encourage low-rank weight matrices. Although these regularizers are convex with respect to their respective arguments, the overall problem remains nonconvex because the hidden-layer outputs depend nonlinearly on the network parameters.

\subsection*{B. Penalty Reformulation}

The objective in \eqref{eq:proposed_P} contains nonsmooth terms composed with nonlinear network mappings, which complicates the direct application of proximal optimization methods. To separate these terms from the network parameters, we introduce auxiliary variables $\{\mathbf{h}_q^{(i)}\}_{q=1,i=1}^{n,l}$ and $\{\Vb^{(j)}\}_{j=1}^{l+1}$, and impose the consistency constraints
$\mathbf{h}_q^{(i)}=\zb_q^{(i)}$ and $\Vb^{(j)}=\Wb^{(j)}$. This gives the equivalent constrained formulation
\small
\begin{subequations}
	\label{eq:constrained}
	\begin{align}
	\min_{\mathcal{X}} \quad
	& \mathcal{L}(\thetab)
	+ \sum_{q=1}^{n}\sum_{i=1}^{l}
	\alpha_i \big\lVert \mathbf{h}_q^{(i)} \big\rVert_1
	+ \sum_{j=1}^{l+1}
	\beta_j \big\lVert \Vb^{(j)} \big\rVert_*,
	\\
	\text{s.t.} \quad
	& \text{C1: } \mathbf{h}_q^{(i)}=\zb_q^{(i)},
	\quad i=1,\dots,l,\quad q=1,\dots,n,
	\\
	& \text{C2: } \Vb^{(j)}=\Wb^{(j)},
	\quad j=1,\dots,l+1,
	\end{align}
\end{subequations}
\normalsize
where
$\mathcal{X}=\{\thetab,\{\mathbf{h}_q^{(i)}\},\{\Vb^{(j)}\}\}$
collects all optimization variables.

We relax the equality constraints by introducing quadratic penalty terms. The resulting penalized problem is
\small
\begin{equation}
\label{eq:penalty}
\begin{split}
\min_{\mathcal{X}} \quad
& \mathcal{L}(\thetab)
+ \sum_{q=1}^{n}\sum_{i=1}^{l}
\left[
\alpha_i \lVert \mathbf{h}_q^{(i)} \rVert_1
+ \frac{1}{2\mu_{h^{(i)}}}
\lVert \zb_q^{(i)}-\mathbf{h}_q^{(i)} \rVert^2
\right]
\\
& + \sum_{j=1}^{l+1}
\left[
\beta_j \lVert \Vb^{(j)} \rVert_*
+ \frac{1}{2\mu_{v^{(j)}}}
\lVert \Wb^{(j)}-\Vb^{(j)} \rVert_F^2
\right],
\end{split}
\end{equation}
\normalsize
where $\mu_{h^{(i)}}>0$ and $\mu_{v^{(j)}}>0$ are penalty parameters. As these parameters approach zero, violations of the corresponding equality constraints are increasingly penalized. Under appropriate existence and regularity conditions, the minimizers of the penalized problem approach solutions of the constrained formulation as the penalty parameters tend to zero.

For the convergence analysis, we additionally restrict the network parameters to an $\ell_\infty$ ball of radius $M$ through the indicator function $r_{\theta}(\thetab)$ defined in Lemma~\ref{lemma:infprox}. The resulting objective can be expressed as
\small
\begin{equation}
\label{eq:final_objective}
\begin{split}
\small \underset{\mathcal{X}}{\mathrm{min}}  \ \ & \underbrace{\mathcal{L}(\thetab) + \sum_{q=1}^{n} \sum_{i=1}^{l}   \frac{1}{2\mu_{ h^{(i)}}} \lVert \mathbf{z}_{q}^{(i)} - \mathbf{h}_{q}^{(i)} \rVert^{2} + \sum_{j=1}^{l+1} \frac{1}{2 \mu_{ v^{(j)}}} \lVert \mathbf{W}^{(j)}  }_\mathrm{I} \\
& \small \underbrace{-\mathbf{V}^{(j)} \rVert_{F}^{2}}_\mathrm{I} \underbrace{+ \sum_{q=1}^{n} \sum_{i=1}^{l} \alpha_{i} \lVert \mathbf{h}_{q}^{(i)} \rVert_{1} + \sum_{j=1}^{l+1} \beta_{j} 
	\lVert \mathbf{V}^{(j)} \rVert_{*} + r_{\theta}(\thetab)}_\mathrm{II}  .
\end{split}
\end{equation}
\normalsize

The objective in \eqref{eq:final_objective} consists of a smooth component, denoted by $g$, and a separable nonsmooth component. Specifically, component I is defined as
\small
\begin{equation}
\label{eq:smooth_part}
\begin{split}
g(\mathbf{h}_{1:n}^{(1:l)},\Vb^{(1:l+1)},\thetab)
= \ & \mathcal{L}(\thetab)
+ \sum_{q=1}^{n}\sum_{i=1}^{l}
\frac{1}{2\mu_{h^{(i)}}}
\lVert \zb_q^{(i)}-\mathbf{h}_q^{(i)} \rVert^2
\\
& + \sum_{j=1}^{l+1}
\frac{1}{2\mu_{v^{(j)}}}
\lVert \Wb^{(j)}-\Vb^{(j)} \rVert_F^2,
\end{split}
\end{equation}
\normalsize
where
$\mathbf{h}_{1:n}^{(1:l)}
=
\{\mathbf{h}_1^{(1)},\dots,\mathbf{h}_1^{(l)},
\dots,\mathbf{h}_n^{(1)},\dots,\mathbf{h}_n^{(l)}\}
$
and
$
\Vb^{(1:l+1)}=\{\Vb^{(1)},\dots,\Vb^{(l+1)}\}.
$

Component II comprises the nonsmooth functions
\begin{equation}
\label{eq:nonsmooth_parts}
r_h(\mathbf{h}_q^{(i)})
= \alpha_i \lVert \mathbf{h}_q^{(i)} \rVert_1,
\qquad
r_v(\Vb^{(j)})
= \beta_j \lVert \Vb^{(j)} \rVert_*,
\end{equation}
together with the parameter constraint $r_{\theta}(\thetab)$. Accordingly, the penalized problem can be compactly written as
\small
\begin{equation}
\label{eq:Q}
\begin{split}
\mathcal{P}_\mu: & \ \arg \min_{\mathcal{X}} \ Q(\mathbf{h}_{1:n}^{(1:l)}, \Vb^{(1:l+1)}, \thetab) = g\big( \mathbf{h}_{1:n,k}^{(1:l)}, \Vb_k^{(1:l+1)}, \thetab \big) \\
& + \sum_{q,i} r_h(\mathbf{h}_q^{(i)}) 
+ \sum_j r_v(\Vb^{(j)}) + r_{\theta}(\thetab).
\end{split}
\end{equation}
\normalsize

\subsection*{C. Alternating Minimization with Closed-Form Proximal Updates}

The reformulated problem $\mathcal{P}_{\mu}$ in \eqref{eq:Q} is nonsmooth and nonconvex. Nevertheless, its variable structure allows the optimization to be carried out through successive block updates. In particular, the auxiliary variables associated with the hidden-layer representations and weight matrices admit closed-form proximal solutions, while the network parameters are updated through a projected proximal-gradient step. We derive these updates below.

\paragraph{Update of $\thetab$: quadratic majorization and extrapolation}

With the auxiliary variables fixed, the subproblem with respect to the network parameters is
\small
\begin{equation}
\label{eq:theta_subproblem}
\thetab_k =
\argmin_{\thetab}
\ g(\mathbf{h}_{1:n,k}^{(1:l)},\Vb_k^{(1:l+1)},\thetab)
+r_{\theta}(\thetab).
\end{equation}
\normalsize
Since $g$ is generally nonconvex in $\thetab$, directly solving this subproblem may be difficult. Instead, we construct a quadratic upper bound using the Lipschitz continuity of its gradient. Specifically, on the compact set
$
\mathcal{C}=\{\thetab:\lVert\thetab\rVert_{\infty}\leq M\},
$
there exists a constant $L_{\theta,k}>0$ such that
\small
\begin{equation}
\label{eq:lipschitz_theta}
\big\lVert
\nabla_{\theta}g(\thetab)
-\nabla_{\theta}g(\thetab')
\big\rVert
\leq
L_{\theta,k}\lVert\thetab-\thetab'\rVert,
\quad
\forall \thetab,\thetab'\in\mathcal{C}.
\end{equation}
\normalsize
This property follows from the smoothness result in Lemma~\ref{lem:compact}, provided that the smooth component is twice continuously differentiable on the parameter domain. Applying the descent lemma at the current iterate $\thetab_{k-1}$ gives
\small
\begin{equation}
\label{eq:descent_g}
\begin{split}
g(\thetab)
\leq \ & g(\thetab_{k-1})
+\nabla_{\theta}^{T}g(\thetab_{k-1})
(\thetab-\thetab_{k-1})
+\frac{L_{\theta,k}}{2}
\lVert\thetab-\thetab_{k-1}\rVert^2.
\end{split}
\end{equation}
\normalsize
Consequently, we minimize the corresponding quadratic surrogate
\small
\begin{equation}
\label{eq:surrogate_theta}
\begin{split}
\widetilde{g}(\thetab)
= \ & g(\thetab_{k-1})
+\nabla_{\theta}^{T}g(\thetab_{k-1})
(\thetab-\thetab_{k-1})
+\frac{1}{2\mu_{\theta,k}}
\lVert\thetab-\thetab_{k-1}\rVert^2,
\end{split}
\end{equation}
\normalsize
where $\mu_{\theta,k}=1/L_{\theta,k}$. The surrogate is strongly convex with respect to $\thetab$, and its minimizer in the presence of the constraint is therefore unique. The resulting update is
\small
\begin{equation}
\label{eq:theta_surrogate}
\thetab_k
=
\argmin_{\thetab}
\ \widetilde{g}(\thetab)+r_{\theta}(\thetab).
\end{equation}
\normalsize
Completing the square and removing terms independent of $\thetab$ yields
\small
\begin{equation}
\label{eq:theta_complete_square}
\begin{split}
\thetab_k
=\argmin_{\thetab}\quad
&\frac{1}{2\mu_{\theta,k}}
\Big\lVert
\thetab-
\big(
\thetab_{k-1}
-\mu_{\theta,k}\nabla_{\theta}g(\thetab_{k-1})
\big)
\Big\rVert^2
+r_{\theta}(\thetab).
\end{split}
\end{equation}
\normalsize
Thus, by the definition of the proximal operator, the parameter update is
\small
\begin{equation}
\label{eq:theta_prox}
\thetab_k
=
\mathrm{Prox}_{\mu_{\theta,k},r_{\theta}}
\Big(
\thetab_{k-1}
-\mu_{\theta,k}\nabla_{\theta}g(\thetab_{k-1})
\Big).
\end{equation}
\normalsize

To accelerate the parameter updates, we incorporate an extrapolation step. Specifically, the extrapolated point is defined as
\small
\begin{equation}
\label{eq:extrapolated_theta}
\widehat{\thetab}_{k-1}
=
\thetab_{k-1}
+\omega_{\theta,k}
(\thetab_{k-1}-\thetab_{k-2}),
\end{equation}
\normalsize
where $\omega_{\theta,k}>0$ denotes the extrapolation weight. The resulting update becomes
\small
\begin{equation}
\label{eq:theta_extrapolated}
\thetab_k
=
\mathrm{Prox}_{\mu_{\theta,k},r_{\theta}}
\Big(
\widehat{\thetab}_{k-1}
-\mu_{\theta,k}
\nabla_{\theta}g(\widehat{\thetab}_{k-1})
\Big).
\end{equation}
\normalsize
To preserve the descent conditions required by the convergence analysis, the extrapolation weight is restricted according to Proposition~\ref{proposition_2}:
\small
\begin{equation}
\label{eq:omega_bound}
0 < \omega_{\theta,k}
\leq
\frac{\gamma-1}{2(\gamma+1)}
\sqrt{
	\delta_k\frac{L_{\theta,k-1}}{L_{\theta,k}}
},
\end{equation}
\normalsize
where $\gamma>1$ and $\delta_k\in(0,1]$.

Since $r_{\theta}$ is the indicator function of the $\ell_\infty$ ball of radius $M$, its proximal operator is the projection onto that ball. Defining
\small
\[
\thetab_k^{\mathrm{grad}}
=
\widehat{\thetab}_{k-1}
-\mu_{\theta,k}
\nabla_{\theta}g(\widehat{\thetab}_{k-1}),
\]
\normalsize
the update admits the component-wise expression
\small
\begin{equation}
\label{eq:theta_closed_form}
[\thetab_k]_i =
\begin{cases}
[\thetab_k^{\mathrm{grad}}]_i,
& |[\thetab_k^{\mathrm{grad}}]_i|\leq M,\\
M\,\mathrm{sign}([\thetab_k^{\mathrm{grad}}]_i),
& |[\thetab_k^{\mathrm{grad}}]_i|>M.
\end{cases}
\end{equation}
\normalsize
Hence, the parameter update consists of a gradient step evaluated at the extrapolated point, followed by projection onto the prescribed bounded parameter set.

\paragraph{Update of $\mathbf{h}_q^{(i)}$}

We next consider the update of an auxiliary variable associated with the $i$-th hidden layer and the $q$-th training sample. Fixing all other blocks, its subproblem is
\small
\begin{equation}
\label{eq:h_subproblem}
\begin{split}
\mathbf{h}_{q,k}^{(i)}
=\argmin_{\mathbf{h}_q^{(i)}}\quad
&g\big(
\mathbf{h}_{1:n,k}^{(<i)},
\mathbf{h}_q^{(i)},
\mathbf{h}_{1:n,k-1}^{(>i)},
\Vb_{k-1}^{(1:l+1)},
\thetab_{k-1}
\big)
\\
&+r_h(\mathbf{h}_q^{(i)}).
\end{split}
\end{equation}
\normalsize
For fixed network parameters, the smooth component is quadratic in $\mathbf{h}_q^{(i)}$. Its gradient is
\small
\begin{equation}
\label{eq:grad_h}
\nabla_{\mathbf{h}_q^{(i)}}g
=
\frac{1}{\mu_{h^{(i)}}}
\big(
\mathbf{h}_q^{(i)}-\zb_q^{(i)}
\big).
\end{equation}
\normalsize
The corresponding gradient is Lipschitz continuous with constant
$L_{h^{(i)}}=1/\mu_{h^{(i)}}$. Taking the proximal step with
$\mu_{h^{(i)}}=1/L_{h^{(i)}}$ gives
\small
\begin{equation}
\label{eq:h_update_detail}
\begin{split}
\mathbf{h}_{q,k}^{(i)}
&=
\mathrm{Prox}_{\mu_{h^{(i)}},\alpha_i\lVert\cdot\rVert_1}
\left(
\mathbf{h}_{q,k-1}^{(i)}
-\mu_{h^{(i)}}
\nabla_{\mathbf{h}_q^{(i)}}g
\right)
\\
&=
\mathrm{Prox}_{\mu_{h^{(i)}},\alpha_i\lVert\cdot\rVert_1}
\left(
\zb_{q,k-1}^{(i)}
\right),
\end{split}
\end{equation}
\normalsize
where the second equality follows by substituting \eqref{eq:grad_h}. Applying the soft-thresholding operator in Lemma~\ref{lemma:l1prox}, the $j$-th component is
\small
\begin{equation}
\label{eq:soft_threshold_detail}
\big[\mathbf{h}_{q,k}^{(i)}\big]_j =
\begin{cases}
\big[\zb_{q,k-1}^{(i)}\big]_j-\alpha_i\mu_{h^{(i)}},
& \big[\zb_{q,k-1}^{(i)}\big]_j>\alpha_i\mu_{h^{(i)}},
\\
0,
& \big|\big[\zb_{q,k-1}^{(i)}\big]_j\big|
\leq\alpha_i\mu_{h^{(i)}},
\\
\big[\zb_{q,k-1}^{(i)}\big]_j+\alpha_i\mu_{h^{(i)}},
& \big[\zb_{q,k-1}^{(i)}\big]_j<-\alpha_i\mu_{h^{(i)}}.
\end{cases}
\end{equation}
\normalsize
Thus, hidden-layer components whose magnitudes do not exceed the threshold $\alpha_i\mu_{h^{(i)}}$ are set to zero, whereas the remaining components are shrunk toward the origin. This operation directly promotes sparsity in the auxiliary representations of the hidden-layer outputs.

\paragraph{Update of $\Vb^{(j)}$}

The auxiliary variables associated with the weight matrices are updated in a similar manner. Fixing all other blocks, the subproblem for $\Vb^{(j)}$ is
\small
\begin{equation}
\label{eq:v_subproblem}
\begin{split}
\Vb_k^{(j)}
=\argmin_{\Vb^{(j)}}\quad
&g\big(
\mathbf{h}_{1:n,k}^{(1:l)},
\Vb_k^{(<j)},
\Vb^{(j)},
\Vb_{k-1}^{(>j)},
\thetab_{k-1}
\big)
+r_v(\Vb^{(j)}).
\end{split}
\end{equation}
\normalsize
The gradient of the smooth component with respect to $\Vb^{(j)}$ is
\small
\begin{equation}
\label{eq:grad_v}
\nabla_{\Vb^{(j)}}g
=
\frac{1}{\mu_{v^{(j)}}}
\big(
\Vb^{(j)}-\Wb^{(j)}
\big),
\end{equation}
\normalsize
which is Lipschitz continuous with constant
$L_{v^{(j)}}=1/\mu_{v^{(j)}}$. Accordingly, the proximal update is
\small
\begin{equation}
\label{eq:v_update_detail}
\begin{split}
\Vb_k^{(j)}
&=
\mathrm{Prox}_{\mu_{v^{(j)}},\beta_j\lVert\cdot\rVert_*}
\left(
\Vb_{k-1}^{(j)}
-\mu_{v^{(j)}}
\nabla_{\Vb^{(j)}}g
\right)
\\
&=
\mathrm{Prox}_{\mu_{v^{(j)}},\beta_j\lVert\cdot\rVert_*}
\left(
\Wb_{k-1}^{(j)}
\right).
\end{split}
\end{equation}
\normalsize
By Lemma~\ref{lemma:nucprox}, this proximal mapping is given by singular value shrinkage. Let
\small
\begin{equation}
\label{eq:svd_W}
\Wb_{k-1}^{(j)}
=
\Ub\,\mathbf{\Sigma}\,\Vb^T,
\qquad
\mathbf{\Sigma}
=
\mathrm{diag}(\sigma_1,\sigma_2,\dots,\sigma_r)
\end{equation}
\normalsize
denote the singular value decomposition of $\Wb_{k-1}^{(j)}$, where
$r=\min(m,n)$ and
$\sigma_1\geq\sigma_2\geq\dots\geq\sigma_r\geq0$.
The resulting update is
\small
\begin{equation}
\label{eq:svd_shrink_detail}
\begin{split}
\Vb_k^{(j)}
&=
\Ub\,
\mathbf{\Sigma}_{\beta_j\mu_{v^{(j)}}}
\,\Vb^T,
\\
\mathbf{\Sigma}_{\beta_j\mu_{v^{(j)}}}
&=
\mathrm{diag}
\left(
\max\{\sigma_i-\beta_j\mu_{v^{(j)}},0\}
\right).
\end{split}
\end{equation}
\normalsize
Therefore, singular values below the threshold
$\beta_j\mu_{v^{(j)}}$ are set to zero, while the remaining singular values are reduced by the same amount. This operation promotes low-rank structure in the auxiliary weight matrices and, through the consistency constraints, encourages low-rank structure in the corresponding network transformations.

\paragraph{DNN-NSR algorithm}

The proposed DNN-NSR method combines the three updates derived above into an iterative block-coordinate scheme. At iteration $k$, the hidden-layer outputs $\zb_{q,k-1}^{(i)}$ are evaluated using the current network parameters $\thetab_{k-1}$. The auxiliary hidden representations are then updated using \eqref{eq:h_update_detail}, followed by the auxiliary weight matrices using \eqref{eq:v_update_detail}. Finally, the network parameters are updated through the extrapolated proximal-gradient step in \eqref{eq:theta_extrapolated}. The resulting update order is
\small \[
\begin{split}
&\thetab_{k-1}
\ \longrightarrow\
\{\mathbf{h}_{q,k}^{(i)}\}_{q=1,i=1}^{n,l}
\longrightarrow\
\{\Vb_k^{(j)}\}_{j=1}^{l+1}
\ \longrightarrow\
\thetab_k.
\end{split}
\]
\normalsize
After completing the prescribed number of iterations, the missing entries are estimated using the output of the trained network:
\small
\begin{equation}
\label{eq:recovery_detail}
x_{ij}=\hat{x}_{ij},
\qquad (i,j)\in\bar{\Omega},
\end{equation}
\normalsize
where $\bar{\Omega}$ denotes the set of unobserved entries.

The proposed updates have several useful properties. First, the auxiliary-variable updates admit closed-form solutions: soft-thresholding for the $\ell_1$ regularizer and singular value shrinkage for the nuclear norm regularizer. The network-parameter update is also available in closed form through projection onto the $\ell_\infty$ ball, following a gradient step. Second, the step size $\mu_{\theta,k}$ is selected according to the reciprocal of the corresponding Lipschitz constant, while the auxiliary-variable updates use the reciprocal Lipschitz constants associated with their quadratic subproblems. Finally, the extrapolation weight is restricted by \eqref{eq:omega_bound} to satisfy the conditions imposed by the convergence analysis. These properties make it possible to handle the nonsmooth regularization terms without resorting to generic nonsmooth optimization solvers.

\subsection*{D. Convergence Analysis}
\begin{algorithm}[H]
	\caption{Deep Neural Network with Nonsmooth Regularization Terms (DNN-NSR) for NLMC.}
	\label{alg:dnnnsr}
	\begin{algorithmic}[1]
		\State \textbf{Input:} $\Omega$, $\Xb$, $M$ in \eqref{eq:indicator}, $\{ \alpha_i \}_{i=1}^{l}$, $\{ \beta_j \}_{j=1}^{l+1}$, number of epochs $K$, number of columns $n$ of $\Xb$, $\lambda$ in \eqref{eq:reg_loss}, $\mu_{h^{(i)}}^{\max}$, $\mu_{h^{(i)}}^{\min}$, $\mu_{v^{(j)}}^{\max}$, $\mu_{v^{(j)}}^{\min}$, $\mu_{\theta}^{\max}$, $\mu_{\theta}^{\min}$, $\gamma = 10^{3}$, scale sets $s_1 = \{0.1, 0.2, 0.8, 0.9\}$, $s_2 = \{0.1, 0.2, 0.5, 0.6\}$, $s_3 = \{1.05, 1.1, 1.2\}$, termination thresholds $\{ \zeta_i \}_{i=1}^{l}$, $\{ \zeta_j \}_{j=1}^{l+1}$, epoch threshold $E$ (e.g., $E=200, 300$), $\delta_{\min} = 0.01$, $\delta_1 = \delta_{\max} = 0.99$.
		\State \textbf{Initialization:} $\thetab_0 = \thetab_{-1}$.
		\State \textbf{Output:} $x_{ij}$ in \eqref{eq:recovery_detail}, neural network parameters $\thetab_K$.
		\algrule
		\State $\mu^{\max} = \mu_{h^{(i)}}^{\max}, \mu_{v^{(j)}}^{\max}, \mu_{\theta}^{\max}$
		\For{$k = 1, 2, \dots, K$}
		\State $\mu_{\theta,k} = \frac{1}{\gamma L_{\theta,k}}$, \quad $\mu_{\theta,0} = \mu_{\theta,1} = \mu_{\theta}^{\max}$;
		\State $\omega_{\theta,1} = \frac{\gamma - 1}{2(\gamma + 1)} \sqrt{\delta_1} \approx 0.5$, \quad $k = 1$;
		\State $\omega_{\theta,k} = \frac{\gamma - 1}{2(\gamma + 1)} \sqrt{\delta_k \frac{L_{\theta,k-1}}{L_{\theta,k}}} \approx 0.5 \sqrt{\delta_k \frac{L_{\theta,k-1}}{L_{\theta,k}}}$, \quad $k \geq 2$;
		\If{the descent condition \eqref{eq:proof_eq54} is not satisfied}
		\State $\mu^{\max} \leftarrow s_1 \cdot \mu^{\max}$;
		\Else
		\For{$q = 1, 2, \dots, n$}
		\State Compute $\{ \mathbf{h}_{q,k}^{(i)} \}_{i=1}^{l}$ using \eqref{eq:h_update_detail};
		\EndFor
		\State Compute $\{ \Vb_k^{(j)} \}_{j=1}^{l+1}$ using \eqref{eq:v_update_detail};
		\State $\widehat{\thetab}_{k-1} = \thetab_{k-1} + \omega_{\theta,k} (\thetab_{k-1} - \thetab_{k-2})$;
		\State Compute $\thetab_k$ using \eqref{eq:theta_extrapolated};
		\State Compute $Q(\xi_k) = Q(\mathbf{h}_{1:n,k}^{(1:l)}, \Vb_k^{(1:l+1)}, \thetab_k)$ using \eqref{eq:Q};
		\State Compute $Q(\xi_{k-1}) = Q(\mathbf{h}_{1:n,k-1}^{(1:l)}, \Vb_{k-1}^{(1:l+1)}, \thetab_{k-1})$ using \eqref{eq:Q};
		\If{$Q(\xi_k) - Q(\xi_{k-1}) > 0$ and $k > E$}
		\State $\delta_k = s_2 \cdot \delta_{k-1}$;
		\If{$\delta_k < \delta_{\min}$}
		\State $\delta_k \leftarrow \delta_{\min}$;
		\EndIf
		\State Repeat lines 5 to 27 with the new $\mu^{\max}$;
		\EndIf
		\If{$Q(\xi_k) - Q(\xi_{k-1}) < 0$ and $k > E$}
		\State $\delta_k = s_3 \cdot \delta_{k-1}$;
		\If{$\delta_k > \delta_{\max}$}
		\State $\delta_k \leftarrow \delta_{\max}$;
		\EndIf
		\State Compute $\mathrm{C1}_k = \| \mathbf{h}_{1:n,k}^{(i)} - \zb_{1:n,k}^{(i)} \|^2 \leq \zeta_i$;
		\State Compute $\mathrm{C2}_k = \| \Vb_k^{(j)} - \Wb_k^{(j)} \|^2 \leq \zeta_j$;
		\If{$k = K$ or ($\mathrm{C1}_k$ and $\mathrm{C2}_k$ satisfied)}
		\State Compute $x_{ij}$ according to \eqref{eq:recovery_detail};
		\Else
		\State Continue until $k = K$ or $\mathrm{C1}_k$ and $\mathrm{C2}_k$ are satisfied;
		\EndIf
		\EndIf
		\EndIf
		\EndFor
	\end{algorithmic}
\end{algorithm}

The proposed objective is nonsmooth and nonconvex, so the appropriate notion of optimality is that of a critical point (Definition~\ref{def:critical}). We establish convergence through three steps: (i) Lipschitz continuity of the partial derivatives of $g$; (ii) a monotone descent property; and (iii) the KL property.

\begin{prepos}
	\label{proposition_1}
	The partial derivatives of $g(\mathbf{h}_{1:n}^{(1:l)}, \Vb^{(1:l+1)}, \thetab)$ with respect to $\mathbf{h}_q^{(i)}$ and $\Vb^{(j)}$ are Lipschitz continuous. Moreover, if all activation functions $g^{(i)}$ are twice continuously differentiable, then the partial derivative of $g$ with respect to $\thetab$ is Lipschitz continuous on the compact set $\mathcal{C} = \{ \thetab : \lVert \thetab \rVert_\infty \le M \}$.
\end{prepos}

\begin{proof}
	The partial derivatives with respect to $\mathbf{h}_q^{(i)}$ and $\Vb^{(j)}$ are
	\small
	\begin{equation}
	\begin{split}
	& \nabla_{h_q^{(i)}} g = \frac{1}{\mu_{h^{(i)},k}} \big( \mathbf{h}_q^{(i)} - \zb_q^{(i)} \big), \ 
	\nabla_{v^{(j)}} g = \frac{1}{\mu_{v^{(j)},k}} \big( \Vb^{(j)} - \Wb^{(j)} \big),
	\end{split}
	\end{equation}
	\normalsize
	which are Lipschitz with moduli $1/\mu_{h^{(i)},k}$ and $1/\mu_{v^{(j)},k}$, respectively. Since all activation functions are twice continuously differentiable, $g$ is twice continuously differentiable in $\thetab$. The set $\mathcal{C}$ is compact, so by Lemma~\ref{lem:compact}, $\nabla_{\theta} g$ is $L$-Lipschitz on $\mathcal{C}$.
\end{proof}

We emphasize that the compactness of $\mathcal{C}$ is essential: without it, $\nabla_{\theta} g$ is not globally Lipschitz, as shown in \cite{mukkamala2019bregman} for the case of linear activation functions. The indicator $r_{\theta}(\thetab)$ enforces this compactness and is therefore not merely a technical device but a structural component of the convergence guarantee.

\begin{prepos}
	\label{proposition_2}
	Let the assumptions of Proposition~\ref{proposition_1} hold. Furthermore, let $\big\{\thetab_k, \{  \mathbf{h}_{q,k}^{(i)} \}, \{ \Vb_k^{(j)} \} \big\}$ be the sequence generated by the update rules in \eqref{eq:theta_extrapolated},  \eqref{eq:h_update_detail}, and \eqref{eq:v_update_detail}. Then, the summable property holds:
\small
	\begin{equation}
	\label{eq:summable}
	\begin{split}
    & \sum_{k=1}^{\infty} \sum_{q=1}^{n} \sum_{i=1}^{l} \nu_{h^{(i)},k} \lVert \mathbf{h}^{(i)}_{q,k} - \mathbf{h}^{(i)}_{q,k-1} \rVert^{2} + \sum_{k=1}^{\infty} \sum_{j=1}^{l+1} \nu_{v^{(j)},k}  
    \lVert \mathbf{V}^{(j)}_{k}  \\
    & - \mathbf{V}^{(j)}_{k-1} \rVert_{F}^{2}
   + \sum_{k=1}^{\infty} \nu_{\theta, k} \lVert \thetab_{k} - \thetab_{k-1} \rVert^{2} \ < \infty,
\end{split}
	\end{equation}
	\normalsize
	with $\nu_{h^{(i)},k} = (\gamma-1) L_{h^{(i)},k}/2$, $\nu_{v^{(j)},k} = (\gamma-1) L_{v^{(j)},k}/2$, and $\nu_{\theta,k} = (1-\delta_{k+1})(\gamma-1) L_{\theta,K}/2$, provided that $\omega_{\theta,k} \le \frac{\gamma-1}{2(\gamma+1)} \sqrt{\delta_k L_{\theta,k-1}/L_{\theta,k}}$ for $\delta_k < 1$ and $\gamma > 1$.
\end{prepos}

\begin{proof}
	Fixing $\mathbf{h}_{1:n,k}^{(1:l)}$ and $\Vb_k^{(1:l+1)}$ and applying Lemma~\ref{lemma:descent}, we can write
	\small
	\begin{equation}
	\label{eq:proof_eq54}
	\begin{split}
	& g \big( \mathbf{h}_{1:n,k}^{(1:l)}, \Vb_k^{(1:l+1)}, \thetab_k \big) \le g \big( \mathbf{h}_{1:n,k}^{(1:l)}, \Vb_k^{(1:l+1)}, \thetab_{k-1} \big) +\\
	& \nabla_{\theta}^{T} g \big( \mathbf{h}_{1:n,k}^{(1:l)}, \Vb_k^{(1:l+1)}, \thetab_{k-1} \big) \big( \thetab_k - \thetab_{k-1} \big) 
	+ \frac{L_{\theta,k}}{2} \lVert \thetab_k - \thetab_{k-1} \rVert^2.
	\end{split}
	\end{equation}
	\normalsize
	Since $\thetab_k$ is the minimizer of \eqref{eq:theta_extrapolated}, we obtain
	\small
	\begin{equation}
	\label{eq:proof_eq55}
	\begin{split}
	& \nabla_{\theta}^{T} g \big( \mathbf{h}_{1:n,k}^{(1:l)}, \Vb_k^{(1:l+1)}, \widehat{\thetab}_{k-1} \big) \big( \thetab_k - \widehat{\thetab}_{k-1} \big) + \frac{1}{2\mu_{\theta,k}} \lVert \thetab_k - 
	\widehat{\thetab}_{k-1} \rVert^2 + \\ & r_{\theta}(\thetab_k) 
	\le \nabla_{\theta}^{T} g \big( \mathbf{h}_{1:n,k}^{(1:l)}, \Vb_k^{(1:l+1)}, \widehat{\thetab}_{k-1} \big) \big( \thetab_{k-1} - 
	\widehat{\thetab}_{k-1} \big) 
	+ \\
	& \frac{1}{2\mu_{\theta,k}} \lVert \thetab_{k-1} - \widehat{\thetab}_{k-1} \rVert^2 + r_{\theta}(\thetab_{k-1}).
	\end{split}
	\end{equation}
	\normalsize
	Summing the inequalities in \eqref{eq:proof_eq54} and \eqref{eq:proof_eq55} and adding and subtracting $\sum_{q=1}^{n} \sum_{i=1}^{l} r_h(\mathbf{h}_{q,k}^{(i)})$ and $\sum_{j=1}^{l+1} r_v(\Vb_k^{(j)})$ to the equation, we get
	\small
	\begin{equation}
	\label{eq:proof_path_1}
	\begin{split}
	& \Big[ g \big( \mathbf{h}_{1:n,k}^{(1:l)}, \Vb_k^{(1:l+1)}, \thetab_{k-1} \big) + r_{\theta}(\thetab_{k-1}) + \sum_{q=1}^{n} \sum_{i=1}^{l} r_h(\mathbf{h}_{q,k}^{(i)})  \\
	& + \sum_{j=1}^{l+1} r_v(\Vb_k^{(j)}) \Big] 
	- \Big[ g \big( \mathbf{h}_{1:n,k}^{(1:l)}, \Vb_k^{(1:l+1)}, \thetab_k \big) + r_{\theta}(\thetab_k) \\
	& + \sum_{q=1}^{n} \sum_{i=1}^{l} r_h(\mathbf{h}_{q,k}^{(i)}) 
	+ \sum_{j=1}^{l+1} r_v(\Vb_k^{(j)}) \Big] \ge \\
	& \underbrace{\nabla_{\theta}^{T} g \big( \mathbf{h}_{1:n,k}^{(1:l)}, \Vb_k^{(1:l+1)}, \widehat{\thetab}_{k-1} \big) \big( \thetab_k - \widehat{\thetab}_{k-1} \big)}_{\mathrm{I}} \\
	&  - \underbrace{\nabla_{\theta}^{T} g \big( \mathbf{h}_{1:n,k}^{(1:l)}, \Vb_k^{(1:l+1)}, \widehat{\thetab}_{k-1} \big) \big( \thetab_{k-1} - \widehat{\thetab}_{k-1} \big)}_{\mathrm{II}} \\
	& + \underbrace{\frac{1}{2\mu_{\theta,k}} \lVert \thetab_k - \widehat{\thetab}_{k-1} \rVert^2 - \frac{1}{2\mu_{\theta,k}} \lVert \thetab_{k-1} - \widehat{\thetab}_{k-1} \rVert^2}_{\mathrm{III}} -\\
	& \nabla_{\theta}^{T} g \big( \mathbf{h}_{1:n,k}^{(1:l)}, \Vb_k^{(1:l+1)}, \thetab_{k-1} \big) \big( \thetab_k - \thetab_{k-1} \big) 
	- \frac{L_{\theta,k}}{2} \lVert \thetab_k - \thetab_{k-1} \rVert^2.
	\end{split}
	\end{equation}
	\normalsize
	Using \eqref{eq:Q}, simplifying parts I and II, and invoking Lemma~\ref{lemma:threepoint} on part III, it holds that
	\small
	\begin{equation}
	\label{eq:proof_path_2}
	\begin{split}
	& Q \big( \mathbf{h}_{1:n,k}^{(1:l)}, \Vb_k^{(1:l+1)}, \thetab_{k-1} \big) - Q \big( \mathbf{h}_{1:n,k}^{(1:l)}, \Vb_k^{(1:l+1)}, \thetab_k \big) \ge \\
	& \quad  \underbrace{\Big[ \nabla_{\theta} g \big( \mathbf{h}_{1:n,k}^{(1:l)}, \Vb_k^{(1:l+1)}, \widehat{\thetab}_{k-1} \big) - \nabla_{\theta} g \big( \mathbf{h}_{1:n,k}^{(1:l)}, \Vb_k^{(1:l+1)}, \thetab_{k-1} \big) \Big]^{T}}_{\mathbf{u1}} \\
	& \quad \underbrace{\big( \thetab_k - \thetab_{k-1} \big)}_{\mathbf{v1}} + \underbrace{\frac{1}{\mu_{\theta,k}} \big( \thetab_{k-1} - \widehat{\thetab}_{k-1} \big)^{T}}_{\mathbf{u2}} \underbrace{\big( \thetab_k - \thetab_{k-1} \big)}_{\mathbf{v2}} \\
	& \quad + \frac{1}{2\mu_{\theta,k}} \lVert \thetab_k - \thetab_{k-1} \rVert^2 - \frac{L_{\theta,k}}{2} \lVert \thetab_k - \thetab_{k-1} \rVert^2.
	\end{split}
	\end{equation}
	\normalsize
	Next, we apply the Cauchy--Schwarz inequality to the parts $\mathbf{u1}$--$\mathbf{v1}$ and $\mathbf{u2}$--$\mathbf{v2}$ in \eqref{eq:proof_path_2}, i.e.,
	\small
	\begin{equation}
	\label{eq:proof_path_3}
	\begin{split}
	& \mathrm{LHS} \ge - \\
	& \Bigg[ \underbrace{\big\lVert \nabla_{\theta} g \big( \mathbf{h}_{1:n,k}^{(1:l)}, \Vb_k^{(1:l+1)}, \widehat{\thetab}_{k-1} \big) - \nabla_{\theta} g \big( \mathbf{h}_{1:n,k}^{(1:l)}, \Vb_k^{(1:l+1)}, \thetab_{k-1} \big) \big\rVert_2}_{\mathrm{IV}}  \\
	& + \frac{1}{\mu_{\theta,k}} \big\lVert \thetab_{k-1} - \widehat{\thetab}_{k-1} \big\rVert \Bigg] \big\lVert \thetab_k - \thetab_{k-1} \big\rVert + \Big( \frac{1}{2\mu_{\theta,k}} - \frac{L_{\theta,k}}{2} \Big) \\
	& \big\lVert \thetab_k - \thetab_{k-1} \big\rVert^2 
	\ge - \Bigg[ \underbrace{L_{\theta,k} \big\lVert \thetab_{k-1} - \widehat{\thetab}_{k-1} \big\rVert}_{V} + \underbrace{\frac{1}{\mu_{\theta,k}} \big\lVert \thetab_{k-1} - \widehat{\thetab}_{k-1} \big\rVert}_{VI} \Bigg] \\
	& \big\lVert \thetab_k - \thetab_{k-1} \big\rVert 
	+ \Big( \frac{1}{2\mu_{\theta,k}} - \frac{L_{\theta,k}}{2} \Big) \big\lVert \thetab_k - \thetab_{k-1} \big\rVert^2.
	\end{split}
	\end{equation}
	\normalsize
	From \eqref{eq:extrapolated_theta}, we obtain $\lVert \thetab_{k-1} - \widehat{\thetab}_{k-1} \rVert = \omega_{\theta,k} \lVert \thetab_{k-1} - \thetab_{k-2} \rVert$, which consequently leads to
	\begin{equation}
	\label{eq:proof_path_5}
	\begin{split}
	\mathrm{LHS} &\ge - \Big( L_{\theta,k} + \frac{1}{\mu_{\theta,k}} \Big) \omega_{\theta,k} \big\lVert \thetab_{k-1} - \thetab_{k-2} \big\rVert \big\lVert \thetab_k - \thetab_{k-1} \big\rVert \\
	& \quad + \Big( \frac{1}{2\mu_{\theta,k}} - \frac{L_{\theta,k}}{2} \Big) \big\lVert \thetab_k - \thetab_{k-1} \big\rVert^2.
	\end{split}
	\end{equation}
	Then, it follows from $\mu_{\theta,k} = 1/(\gamma L_{\theta,k})$ for $\gamma > 1$ that
	\begin{equation}
	\label{eq:proof_first_part}
	\begin{split}
	\mathrm{LHS} &\ge - (\gamma + 1) L_{\theta,k} \omega_{\theta,k} \big\lVert \thetab_{k-1} - \thetab_{k-2} \big\rVert \big\lVert \thetab_k - \thetab_{k-1} \big\rVert \\
	& \quad + \frac{(\gamma - 1) L_{\theta,k}}{2} \big\lVert \thetab_k - \thetab_{k-1} \big\rVert^2.
	\end{split}
	\end{equation}
	On the other hand, setting $a = (\gamma + 1) L_{\theta,k} \omega_{\theta,k} \lVert \thetab_{k-1} - \thetab_{k-2} \rVert$, $b = \lVert \thetab_k - \thetab_{k-1} \rVert$, and $\epsilon = (\gamma - 1) L_{\theta,k}/2$ and applying Lemma~\ref{lemma:young} imply
	\begin{equation}
	\label{eq:proof_second_part}
	\begin{split}
	& Q \big( \mathbf{h}_{1:n,k}^{(1:l)}, \Vb_k^{(1:l+1)}, \thetab_{k-1} \big) - Q \big( \mathbf{h}_{1:n,k}^{(1:l)}, \Vb_k^{(1:l+1)}, \thetab_k \big) \\
	& \quad \ge - \frac{(\gamma + 1)^2 L_{\theta,k}}{\gamma - 1} \omega_{\theta,k}^2 \big\lVert \thetab_{k-1} - \thetab_{k-2} \big\rVert^2 \\
	& \quad + \frac{(\gamma - 1) L_{\theta,k}}{4} \big\lVert \thetab_k - \thetab_{k-1} \big\rVert^2.
	\end{split}
	\end{equation}
	Following the same procedure as in \eqref{eq:proof_eq54} and \eqref{eq:proof_eq55} for the subproblems \eqref{eq:h_subproblem} and \eqref{eq:v_subproblem}, we end up with
	\begin{equation}
	\label{eq:proof_h_sub}
	\begin{split}
	& g \big( \mathbf{h}_{1:q-1,k}^{(1:l)}, \mathbf{h}_{q,k}^{(<i)}, \mathbf{h}_{q,k}^{(i)}, \mathbf{h}_{q,k-1}^{(>i)}, \mathbf{h}_{q+1:n,k-1}^{(1:l)}, \Vb_{k-1}^{(1:l+1)}, \thetab_{k-1} \big) \\
	& \quad + r_h(\mathbf{h}_{q,k}^{(i)}) 
	+ a_i \big\lVert \mathbf{h}_{q,k}^{(i)} - \mathbf{h}_{q,k-1}^{(i)} \big\rVert^2 
	\le g \big( \mathbf{h}_{1:q-1,k}^{(1:l)}, \mathbf{h}_{q,k}^{(<i)}, \\
	& \mathbf{h}_{q,k-1}^{(i)}, \mathbf{h}_{q,k-1}^{(>i)}, \mathbf{h}_{q+1:n,k-1}^{(1:l)}, \Vb_{k-1}^{(1:l+1)}, \thetab_{k-1} \big) + r_h(\mathbf{h}_{q,k-1}^{(i)})
	\end{split}
	\end{equation}
	and
	\begin{equation}
	\label{eq:proof_v_sub}
	\begin{split}
	& g \big( \mathbf{h}_{1:n,k}^{(1:l)}, \Vb_k^{(<j)}, \Vb_k^{(j)}, \Vb_{k-1}^{(>j)}, \thetab_{k-1} \big) + r_v(\Vb_k^{(j)}) + b_j \big\lVert \Vb_k^{(j)} - \\
	& \quad  \Vb_{k-1}^{(j)} \big\rVert_F^2 
	\le g \big( \mathbf{h}_{1:n,k}^{(1:l)}, \Vb_k^{(<j)}, \Vb_{k-1}^{(j)}, \Vb_{k-1}^{(>j)}, \thetab_{k-1} \big) \\
	& \quad + r_v(\Vb_{k-1}^{(j)}),
	\end{split}
	\end{equation}
	where the identities
	\begin{equation}
	\begin{split}
	& a_i = \Big( \frac{1}{2\mu_{h^{(i)},k}} - \frac{L_{h^{(i)},k}}{2} \Big) = \frac{(\gamma - 1) L_{h^{(i)},k}}{2} > 0, \\ 
	& b_j = \Big( \frac{1}{2\mu_{v^{(j)},k}} - \frac{L_{v^{(j)},k}}{2} \Big) = \frac{(\gamma - 1) L_{v^{(j)},k}}{2} > 0
	\end{split}
	\end{equation}
	hold for $\gamma > 1$, since $\mu_{h^{(i)},k} = 1/(\gamma L_{h^{(i)},k})$ and $\mu_{v^{(j)},k} = 1/(\gamma L_{v^{(j)},k)}$. Now, fixing $\thetab$, writing \eqref{eq:proof_h_sub} and \eqref{eq:proof_v_sub} for $\mathbf{h}_1^{(1)}, \dots, \mathbf{h}_n^{(l)}$ and $\Vb^{(1)}, \dots, \Vb^{(l+1)}$, respectively, and adding all the resulting inequalities, we consequently come to
	\begin{equation}
	\label{eq:proof_eq56}
	\begin{split}
	& g(\mathbf{h}_{1:n,k}^{(1:l)}, \Vb_k^{(1:l+1)}, \thetab_{k-1}) + \sum_{q=1}^{n} \sum_{i=1}^{l} r_h(\mathbf{h}_{q,k}^{(i)}) + \sum_{j=1}^{l+1} r_v(\Vb_k^{(j)}) \\
	& \quad + \sum_{q=1}^{n} \sum_{i=1}^{l} a_i \big\lVert \mathbf{h}_{q,k}^{(i)} - \mathbf{h}_{q,k-1}^{(i)} \big\rVert^2 + \sum_{j=1}^{l+1} b_j \big\lVert \Vb_k^{(j)} - \Vb_{k-1}^{(j)} \big\rVert_F^2 \\ 
	& \quad + r_{\theta}(\thetab_{k-1}) \le g(\mathbf{h}_{1:n,k-1}^{(1:l)}, \Vb_{k-1}^{(1:l+1)}, \thetab_{k-1}) \\
	& \quad + \sum_{q=1}^{n} \sum_{i=1}^{l} r_h(\mathbf{h}_{q,k-1}^{(i)}) + \sum_{j=1}^{l+1} r_v(\Vb_{k-1}^{(j)}) + r_{\theta}(\thetab_{k-1}).
	\end{split}
	\end{equation}
	Adding $r_{\theta}(\thetab_{k-1})$ to both sides of \eqref{eq:proof_eq56} and using \eqref{eq:Q}, the equation \eqref{eq:proof_eq56} can be rewritten as
	\begin{equation}
	\label{eq:proof_eq56_1}
	\begin{split}
	& Q(\mathbf{h}_{1:n,k-1}^{(1:l)}, \Vb_{k-1}^{(1:l+1)}, \thetab_{k-1}) - Q(\mathbf{h}_{1:n,k}^{(1:l)}, \Vb_k^{(1:l+1)}, \thetab_{k-1}) \\
	& \quad \ge \sum_{q=1}^{n} \sum_{i=1}^{l} a_i \big\lVert \mathbf{h}_{q,k}^{(i)} - \mathbf{h}_{q,k-1}^{(i)} \big\rVert^2 + \sum_{j=1}^{l+1} b_j \big\lVert \Vb_k^{(j)} - \Vb_{k-1}^{(j)} \big\rVert_F^2.
	\end{split}
	\end{equation}
	Summing up both sides of \eqref{eq:proof_second_part} and \eqref{eq:proof_eq56_1} leads to
	\small
	\begin{equation}
	\label{eq:proof_eq57}
	\begin{split}
	& Q(\mathbf{h}_{1:n,k-1}^{(1:l)}, \Vb_{k-1}^{(1:l+1)}, \thetab_{k-1}) - Q(\mathbf{h}_{1:n,k}^{(1:l)}, \Vb_k^{(1:l+1)}, \thetab_k) \ge \\
	&   - \frac{(\gamma + 1)^2 L_{\theta,k}}{\gamma - 1} \omega_{\theta,k}^2 \big\lVert \thetab_{k-1} - \thetab_{k-2} \big\rVert^2 + \frac{(\gamma - 1) L_{\theta,k}}{4} \big\lVert \thetab_k - \thetab_{k-1} \big\rVert^2  \\
	& + \sum_{q=1}^{n} \sum_{i=1}^{l} a_i \big\lVert \mathbf{h}_{q,k}^{(i)} - \mathbf{h}_{q,k-1}^{(i)} \big\rVert^2 + \sum_{j=1}^{l+1} b_j \big\lVert \Vb_k^{(j)} - \Vb_{k-1}^{(j)} \big\rVert_F^2. 
	\end{split}
	\end{equation}
	\normalsize
	By using the upper bound $\omega_{\theta,k} \le \frac{\gamma - 1}{2(\gamma + 1)} \sqrt{\delta_k L_{\theta,k-1}/L_{\theta,k}}$ for $\delta_k < 1$, inequality \eqref{eq:proof_eq57} can be rewritten as
	\small
	\begin{equation}
	\label{eq:proof_eq57_edited}
	\begin{split}
	& Q(\mathbf{h}_{1:n,k-1}^{(1:l)}, \Vb_{k-1}^{(1:l+1)}, \thetab_{k-1}) - Q(\mathbf{h}_{1:n,k}^{(1:l)}, \Vb_k^{(1:l+1)}, \thetab_k) \ge \\
	&   - \frac{\delta_k (\gamma - 1) L_{\theta,k-1}}{4} \big\lVert \thetab_{k-1} - \thetab_{k-2} \big\rVert^2  + \frac{(\gamma - 1) L_{\theta,k}}{4} \big\lVert \thetab_k - \thetab_{k-1} \big\rVert^2  \\
	&  + \sum_{q=1}^{n} \sum_{i=1}^{l} a_i \big\lVert \mathbf{h}_{q,k}^{(i)} - \mathbf{h}_{q,k-1}^{(i)} \big\rVert^2 + \sum_{j=1}^{l+1} b_j \big\lVert \Vb_k^{(j)} - \Vb_{k-1}^{(j)} \big\rVert_F^2.
	\end{split}
	\end{equation}
	\normalsize
	Summing up both sides of this inequality from $1$ to $K$, it holds that
	\small
	\begin{equation}
	\label{eq:proof_firstpart}
	\begin{split}
	& Q(\mathbf{h}_{1:n,0}^{(1:l)}, \Vb_0^{(1:l+1)}, \thetab_0) - Q(\mathbf{h}_{1:n,K}^{(1:l)}, \Vb_K^{(1:l+1)}, \thetab_K) \\
	&  \ge \sum_{k=1}^{K-1} \frac{(1 - \delta_{k+1})(\gamma - 1) L_{\theta,k}}{4} \big\lVert \thetab_k - \thetab_{k-1} \big\rVert^2 \\
	&  + \frac{(\gamma - 1) L_{\theta,K}}{4} \big\lVert \thetab_K - \thetab_{K-1} \big\rVert^2 
	- \frac{\delta_1 (\gamma - 1) L_{\theta,0}}{4} \big\lVert \thetab_0 - \thetab_{-1} \big\rVert^2 \\
	&  + \sum_{k=1}^{K} \sum_{q=1}^{n} \sum_{i=1}^{l} a_i \big\lVert \mathbf{h}_{q,k}^{(i)} - \mathbf{h}_{q,k-1}^{(i)} \big\rVert^2 
    + \sum_{k=1}^{K} \sum_{j=1}^{l+1} b_j \big\lVert \Vb_k^{(j)} - \Vb_{k-1}^{(j)} \big\rVert_F^2.
	\end{split}
	\end{equation}
	\normalsize
	Since $\thetab_0 = \thetab_{-1}$, we end up with
	\small
	\begin{equation}
	\label{eq:proof_firstpart_1}
	\begin{split}
	& Q(\mathbf{h}_{1:n,0}^{(1:l)}, \Vb_0^{(1:l+1)}, \thetab_0) - Q(\mathbf{h}_{1:n,K}^{(1:l)}, \Vb_K^{(1:l+1)}, \thetab_K) \\
	&  \ge \sum_{k=1}^{K} \frac{(1 - \delta_{k+1})(\gamma - 1) L_{\theta,k}}{4} \big\lVert \thetab_k - \thetab_{k-1} \big\rVert^2 \\
	&  + \sum_{k=1}^{K} \sum_{q=1}^{n} \sum_{i=1}^{l} a_i \big\lVert \mathbf{h}_{q,k}^{(i)} - \mathbf{h}_{q,k-1}^{(i)} \big\rVert^2 
	+ \sum_{k=1}^{K} \sum_{j=1}^{l+1} b_j \big\lVert \Vb_k^{(j)} - \Vb_{k-1}^{(j)} \big\rVert_F^2,
	\end{split}
	\end{equation}
	\normalsize
	which guarantees
	\begin{equation}
	\label{eq:proof_finalformula}
	Q(\mathbf{h}_{1:n,K}^{(1:l)}, \Vb_K^{(1:l+1)}, \thetab_K) \le Q(\mathbf{h}_{1:n,0}^{(1:l)}, \Vb_0^{(1:l+1)}, \thetab_0).
	\end{equation}
	Letting $K \to \infty$ in this inequality, it can be deduced
	\begin{equation}
	\label{eq:proof_eq64}
	\begin{split}
	& \sum_{k=1}^{\infty} \sum_{q=1}^{n} \sum_{i=1}^{l} a_i \big\lVert \mathbf{h}_{q,k}^{(i)} - \mathbf{h}_{q,k-1}^{(i)} \big\rVert^2 + \sum_{k=1}^{\infty} \sum_{j=1}^{l+1} b_j \big\lVert \Vb_k^{(j)} - \Vb_{k-1}^{(j)} \big\rVert_F^2 \\
	& \quad + \sum_{k=1}^{\infty} \frac{(1 - \delta_{k+1})(\gamma - 1) L_{\theta,k}}{4} \big\lVert \thetab_k - \thetab_{k-1} \big\rVert^2 < \infty,
	\end{split}
	\end{equation}
	which leads to our desired result.
\end{proof}

Proposition~\ref{proposition_2} implies that the objective $Q$ is monotonically non-increasing along the iterates, and that the successive differences of all variables are square-summable. This is the key ingredient for establishing the boundedness of the iterates and the convergence of the algorithm. In what follows, we establish the remaining steps of the convergence analysis: boundedness of the iterates, existence of limit points, convergence of the regularizers, criticality of limit points, and finally convergence of the entire sequence via the KL property.

\subsubsection*{Boundedness of the Iterates}

Before establishing convergence, we need to ensure that the sequence generated by Algorithm~\ref{alg:dnnnsr} does not escape to infinity. This is the content of the following proposition.

\begin{prepos}
	\label{proposition_3}
	If the function $Q$ given in \eqref{eq:Q} has bounded lower level sets, then the sequence $\big\{ \{ \mathbf{h}_k^{(i)} \}_{i=1}^{l}, \{ \Vb_k^{(j)} \}_{j=1}^{l+1}, \thetab_k \big\}$ generated by Algorithm~\ref{alg:dnnnsr} is bounded.
\end{prepos}

\begin{proof}
	It follows from the monotone descent property \eqref{eq:proof_finalformula} established in Proposition~\ref{proposition_2} that, for every $K \ge 0$,
	\begin{equation}
	Q \big( \mathbf{h}_{1:n,K}^{(1:l)}, \Vb_K^{(1:l+1)}, \thetab_K \big) \le Q \big( \mathbf{h}_{1:n,0}^{(1:l)}, \Vb_0^{(1:l+1)}, \thetab_0 \big).
	\end{equation}
	Therefore, the entire sequence is contained in the sublevel set
	\begin{equation}
	\label{eq:sublevel_set}
	\begin{split}
	& \mathcal{L}_Q \big( \mathbf{h}_0^{(1:l)}, \Vb_0^{(1:l+1)}, \thetab_0 \big) \\
	& \quad = \Big\{ \big( \{ \mathbf{h}_k^{(i)} \}_{i=1}^{l}, \{ \Vb_k^{(j)} \}_{j=1}^{l+1}, \thetab_k \big) : Q \big( \{ \mathbf{h}_k^{(i)} \}, \{ \Vb_k^{(j)} \}, \thetab_k \big) \le \\
	& \quad Q \big( \mathbf{h}_0, \Vb_0, \thetab_0 \big) \Big\}.
	\end{split}
	\end{equation}
	Formally,
	\begin{equation}
	\begin{split}
	& \big\{ \{ \mathbf{h}_k^{(i)} \}_{i=1}^{l}, \{ \Vb_k^{(j)} \}_{j=1}^{l+1}, \thetab_k \big\} \subseteq \mathcal{L}_Q \big( \mathbf{h}_0^{(1:l)}, \Vb_0^{(1:l+1)}, \thetab_0 \big), \\
	& \quad \forall k \ge 0.
	\end{split}
	\end{equation}
	Since $\mathcal{L}_Q$ is bounded by assumption, the sequence is bounded. This completes the proof.
\end{proof}

We note that the boundedness of the sublevel sets of $Q$ is guaranteed by the indicator $r_{\theta}(\thetab)$ for the parameter block $\thetab$, and by the fact that the penalty terms $\lVert \zb_q^{(i)} - \mathbf{h}_q^{(i)} \rVert^2$ and $\lVert \Wb^{(j)} - \Vb^{(j)} \rVert_F^2$ grow quadratically in $\mathbf{h}_q^{(i)}$ and $\Vb^{(j)}$. Consequently, the assumption of Proposition~\ref{proposition_3} is satisfied in our setting.

\subsubsection*{Existence of Limit Points and Convergence of the Regularizers}

Having established boundedness, we now show that the sequence has at least one limit point, and that the nonsmooth regularizers converge along any convergent subsequence.

\begin{theorem}[Subsequence convergence]
	\label{theorem:subseq}
	Let $\big\{ \{ \mathbf{h}_k^{(i)} \}_{i=1}^{l}, \{ \Vb_k^{(j)} \}_{j=1}^{l+1}, \thetab_k \big\}$ be a bounded sequence generated by Algorithm~\ref{alg:dnnnsr}. Then:
	\begin{enumerate}
		\item[(a)] The sequence has at least one finite limit point, say $\big\{ \{ \bar{\mathbf{h}}^{(i)} \}_{i=1}^{l}, \{ \bar{\Vb}^{(j)} \}_{j=1}^{l+1}, \bar{\thetab} \big\}$.
		\item[(b)] The regularizers converge along any convergent subsequence:
		\begin{equation}
		\label{eq:regularizer_convergence}
		\begin{split}
		\lim_{\mathcal{I} \ni k \to \infty} r_h(\mathbf{h}_k^{(i)}) &= r_h(\bar{\mathbf{h}}^{(i)}), \quad i=1,\dots,l, \\
		\lim_{\mathcal{I} \ni k \to \infty} r_v(\Vb_k^{(j)}) &= r_v(\bar{\Vb}^{(j)}), \quad j=1,\dots,l+1, \\
		\lim_{\mathcal{I} \ni k \to \infty} r_{\theta}(\thetab_k) &= r_{\theta}(\bar{\thetab}),
		\end{split}
		\end{equation}
		where $\mathcal{I}$ denotes the subsequence index set.
		\item[(c)] Every limit point of the sequence is a critical point of $Q$ in \eqref{eq:Q}.
		\item[(d)] If a subsequence $\big\{ \{ \mathbf{h}_k^{(i)} \}, \{ \Vb_k^{(j)} \}, \thetab_k \big\}_{k \in \mathcal{I}}$ converges to $\big\{ \{ \bar{\mathbf{h}}^{(i)} \}, \{ \bar{\Vb}^{(j)} \}, \bar{\thetab} \big\}$, then
		\begin{equation}
		\label{eq:Q_convergence_subseq}
		\lim_{\mathcal{I} \ni k \to \infty} Q \big( \mathbf{h}_k^{(1:l)}, \Vb_k^{(1:l+1)}, \thetab_k \big) = Q \big( \bar{\mathbf{h}}^{(1:l)}, \bar{\Vb}^{(1:l+1)}, \bar{\thetab} \big).
		\end{equation}
	\end{enumerate}
\end{theorem}

\begin{proof}
	\textbf{Part (a).} By Proposition~\ref{proposition_3}, the sequence is bounded. The Bolzano--Weierstrass theorem \cite{Bolzano} guarantees that any bounded sequence in $\mathbb{R}^d$ contains a convergent subsequence. Hence, there exists a subsequence indexed by $\mathcal{I}$ and a limit point $\big\{ \{ \bar{\mathbf{h}}^{(i)} \}, \{ \bar{\Vb}^{(j)} \}, \bar{\thetab} \big\}$ such that
	\begin{equation}
	\lim_{\mathcal{I} \ni k \to \infty} \mathbf{h}_k^{(i)} = \bar{\mathbf{h}}^{(i)}, \quad \lim_{\mathcal{I} \ni k \to \infty} \Vb_k^{(j)} = \bar{\Vb}^{(j)}, \quad \lim_{\mathcal{I} \ni k \to \infty} \thetab_k = \bar{\thetab}.
	\end{equation}
	
	\textbf{Part (b).} We prove the result for $r_{\theta}$; the proofs for $r_h$ and $r_v$ are analogous. Let $\bar{\thetab}$ be a limit point of $\{ \thetab_k \}$, and let $\{ \thetab_k \}_{k \in \mathcal{I}}$ be a subsequence converging to $\bar{\thetab}$. From Proposition~\ref{proposition_2}, we have $\lVert \thetab_{k+1} - \thetab_k \rVert \to 0$, which implies that for any fixed $\gamma \ge 0$, the shifted sequence $\{ \thetab_{k+\gamma} \}_{k \in \mathcal{I}}$ also converges to $\bar{\thetab}$. It follows from the optimality condition \eqref{eq:theta_extrapolated} that
	\begin{equation}
	\label{eq:proof_theta_opt}
	\begin{split}
	& \nabla_{\theta}^{T} g \big( \mathbf{h}_k^{(1:l)}, \Vb_k^{(1:l+1)}, \thetab_k \big) \big( \thetab_k - \widehat{\thetab}_{k-1} \big) + \frac{1}{2\mu_{\theta,k}} \big\lVert \thetab_k - \widehat{\thetab}_{k-1} \big\rVert^2 \\
	& \quad + r_{\theta}(\thetab_k) 
	\le \nabla_{\theta}^{T} g \big( \mathbf{h}_k^{(1:l)}, \Vb_k^{(1:l+1)}, \widehat{\thetab}_{k-1} \big) \big( \thetab - \widehat{\thetab}_{k-1} \big)\\
	& \quad + \frac{1}{2\mu_{\theta,k}} \big\lVert \thetab - \widehat{\thetab}_{k-1} \big\rVert^2 + r_{\theta}(\thetab),
	\end{split}
	\end{equation}
	which holds for every $\thetab$. Taking the limit along the subsequence $\mathcal{I} \ni k \to \infty$, the terms involving the gradient and the quadratic penalty tend to zero, because $\thetab_k - \widehat{\thetab}_{k-1} \to 0$ and $\mu_{\theta,k}$ is bounded away from zero on the compact set $\mathcal{C}$. We obtain
	\begin{equation}
	\label{eq:proof_theta_limit}
	\begin{split}
	\lim_{\mathcal{I} \ni k \to \infty} r_{\theta}(\thetab_k) &\le \nabla_{\theta}^{T} g \big( \bar{\mathbf{h}}^{(1:l)}, \bar{\Vb}^{(1:l+1)}, \bar{\thetab} \big) \big( \thetab - \bar{\thetab} \big) \\
	& \quad + \frac{1}{2 \bar{\mu}_{\theta}} \big\lVert \thetab - \bar{\thetab} \big\rVert^2 + r_{\theta}(\thetab),
	\end{split}
	\end{equation}
	where $\bar{\mu}_{\theta}$ is a constant in the interval $(0, 1/\bar{L}_{\theta}]$, and $\bar{L}_{\theta}$ is the Lipschitz constant of $\nabla_{\theta} g$ at $\thetab = \bar{\thetab}$ (which exists by Proposition~\ref{proposition_1}). Setting $\thetab = \bar{\thetab}$ in \eqref{eq:proof_theta_limit} yields
	\begin{equation}
	\label{eq:proof_theta_lsc_upper}
	\lim_{\mathcal{I} \ni k \to \infty} r_{\theta}(\thetab_k) \le r_{\theta}(\bar{\thetab}).
	\end{equation}
	On the other hand, $r_{\theta}$ is the indicator of a closed set, hence lower semi-continuous \cite{amini2019}, which gives
	\begin{equation}
	\label{eq:proof_theta_lsc_lower}
	\lim_{\mathcal{I} \ni k \to \infty} r_{\theta}(\thetab_k) \ge r_{\theta}(\bar{\thetab}).
	\end{equation}
	Combining \eqref{eq:proof_theta_lsc_upper} and \eqref{eq:proof_theta_lsc_lower}, we obtain
	\begin{equation}
	\lim_{\mathcal{I} \ni k \to \infty} r_{\theta}(\thetab_k) = r_{\theta}(\bar{\thetab}).
	\end{equation}
	Since the $\ell_1$ norm and the nuclear norm are also lower semi-continuous \cite{norm1}, the same argument applies to $r_h$ and $r_v$. This proves \eqref{eq:regularizer_convergence}.
	
	\textbf{Part (c).} By \eqref{eq:proof_theta_limit} and Part (b), we can write
	\begin{equation}
	\begin{split}
	\label{eq:proof_critical}
	\bar{\thetab} = \argmin_{\thetab} \ & \nabla_{\theta}^{T} g \big( \bar{\mathbf{h}}^{(1:l)}, \bar{\Vb}^{(1:l+1)}, \bar{\thetab} \big) \big( \thetab - \bar{\thetab} \big) \\
	& + \frac{1}{2 \bar{\mu}_{\theta}} \big\lVert \thetab - \bar{\thetab} \big\rVert^2 
	+ r_{\theta}(\thetab).
	\end{split}
	\end{equation}
	The first-order optimality condition for \eqref{eq:proof_critical} at $\thetab = \bar{\thetab}$ is
	\begin{equation}
	\label{eq:proof_fermat}
	\mathbf{0} \in \nabla_{\theta} g \big( \bar{\thetab} \big) + \partial r_{\theta}(\bar{\thetab}),
	\end{equation}
	where $\partial r_{\theta}(\bar{\thetab})$ is the limiting subdifferential of $r_{\theta}$ at $\bar{\thetab}$. By Definition~\ref{def:critical}, this means that $\bar{\thetab}$ is a critical point of $Q$ with respect to $\thetab$. The same argument applies to $\mathbf{h}^{(1:l)}$ and $\Vb^{(1:l+1)}$, so every limit point of the sequence is a critical point of $Q$.
	
	\textbf{Part (d).} By Part (b) and the continuity of $g$ (which follows from the smoothness of $g$ on the compact set $\mathcal{C}$), we can write
	\begin{equation}
	\label{eq:proof_Q_convergence}
	\begin{split}
	& \lim_{\mathcal{I} \ni k \to \infty} Q \big( \mathbf{h}_k^{(1:l)}, \Vb_k^{(1:l+1)}, \thetab_k \big) \\
	&= \lim_{\mathcal{I} \ni k \to \infty} g \big( \mathbf{h}_k^{(1:l)}, \Vb_k^{(1:l+1)}, \thetab_k \big) + \lim_{\mathcal{I} \ni k \to \infty} r_{\theta}(\thetab_k) \\
	&= g \big( \bar{\mathbf{h}}^{(1:l)}, \bar{\Vb}^{(1:l+1)}, \bar{\thetab} \big) + r_{\theta}(\bar{\thetab}) \\
	&= Q \big( \bar{\mathbf{h}}^{(1:l)}, \bar{\Vb}^{(1:l+1)}, \bar{\thetab} \big),
	\end{split}
	\end{equation}
	where we have used the convergence of the regularizers established in Part (b). This completes the proof.
\end{proof}

\subsubsection*{The KL Property}

Theorem~\ref{theorem:subseq} establishes that every limit point of the sequence is a critical point of $Q$. To upgrade this subsequential result to convergence of the entire sequence, we invoke the KL property. The following proposition verifies that $Q$ belongs to the class of KL functions.

\begin{prepos}
	\label{proposition:kl}
	The objective function $Q \big( \mathbf{h}^{(1:l)}, \Vb^{(1:l+1)}, \thetab \big)$ defined in \eqref{eq:Q} is a KL function.
\end{prepos}

\begin{proof}
	As stated in \eqref{eq:Q}, $Q$ is the sum of four terms: the smooth function $g$, the $\ell_1$ norm $r_h$, the nuclear norm $r_v$, and the indicator $r_{\theta}$. We verify the KL property of each term.
	
	First, $g$ is twice continuously differentiable on the compact set $\mathcal{C}$ (by Proposition~\ref{proposition_1}), and it was shown in \cite{amini2019} that such functions satisfy the KL property. Second, the indicator $r_{\theta}$ of a compact set satisfies the KL property with exponent $\theta = 0$, as established in \cite{amini2019}. Third, the $\ell_1$ norm and the nuclear norm are both semialgebraic functions, and it is a classical result \cite{KL_Norm1} that semialgebraic functions satisfy the KL property. Finally, the KL property is preserved under finite sums: if $f_1, \dots, f_m$ satisfy the KL property at $\bar{\xb}$, then so does $f_1 + \dots + f_m$ \cite{KL_property}. Applying this to the four terms of $Q$, we conclude that $Q$ satisfies the KL property at every point of its domain.
\end{proof}

\subsubsection*{Global Convergence}

With the KL property established, we can now invoke the general convergence theorem of \cite{whole_sequence} to obtain convergence of the entire sequence.

\begin{theorem}[Global convergence \cite{whole_sequence}]
	\label{theorem:global}
	Consider problem \eqref{eq:Q} where $Q \big( \mathbf{h}^{(1:l)}, \Vb^{(1:l+1)}, \thetab \big)$ is proper and bounded on its domain, $g \big( \mathbf{h}^{(1:l)}, \Vb^{(1:l+1)}, \thetab \big)$ is continuously differentiable, $\{ r_h(\mathbf{h}^{(i)}) \}_{i=1}^{l}$, $\{ r_v(\Vb^{(j)}) \}_{j=1}^{l+1}$, and $r_{\theta}(\thetab)$ are proper and lower semi-continuous, $g$ has Lipschitz gradient with respect to all variables, and the problem has a critical point $\mathbf{x}^{*}$ satisfying $\mathbf{0} \in \partial Q(\mathbf{x}^{*})$. Let $\big\{ \mathbf{h}^{(1:l)}(k), \Vb^{(1:l+1)}(k), \thetab(k) \big\}_{k \ge 1}$ be generated by Algorithm~\ref{alg:dnnnsr}. Assume:
	\begin{enumerate}
		\item[(a)] the sequence $\big\{ \mathbf{h}_k^{(1:l)}, \Vb_k^{(1:l+1)}, \thetab_k \big\}_{k \ge 1}$ has a finite limit point $\big\{ \bar{\mathbf{h}}^{(1:l)}, \bar{\Vb}^{(1:l+1)}, \bar{\thetab} \big\}$;
		\item[(b)] $Q$ is a KL function.
	\end{enumerate}
	Then the entire sequence converges to $\big\{ \bar{\mathbf{h}}^{(1:l)}, \bar{\Vb}^{(1:l+1)}, \bar{\thetab} \big\}$:
	\begin{equation}
	\lim_{k \to \infty} \big\{ \mathbf{h}_k^{(1:l)}, \Vb_k^{(1:l+1)}, \thetab_k \big\} = \big\{ \bar{\mathbf{h}}^{(1:l)}, \bar{\Vb}^{(1:l+1)}, \bar{\thetab} \big\}.
	\end{equation}
\end{theorem}

\begin{proof}
	We verify that all assumptions of Theorem~2.7 of \cite{whole_sequence} are satisfied. First, $Q$ is proper and bounded on its domain by construction: it is a sum of the smooth loss $\mathcal{L}(\thetab)$, quadratic penalty terms, and the nonsmooth regularizers, all of which are proper. Second, $g$ is continuously differentiable with Lipschitz gradient on the compact set $\mathcal{C}$, as established in Proposition~\ref{proposition_1}. Third, $r_h$, $r_v$, and $r_{\theta}$ are proper and lower semi-continuous, as required. Fourth, the sequence has a finite limit point by Theorem~\ref{theorem:subseq}(a). Fifth, $Q$ is a KL function by Proposition~\ref{proposition:kl}. Sixth, the successive differences of all variables are square-summable by Proposition~\ref{proposition_2}, which is the key technical condition needed for the theorem. Under these conditions, Theorem~2.7 of \cite{whole_sequence} guarantees that the entire sequence converges to the limit point $\big\{ \bar{\mathbf{h}}^{(1:l)}, \bar{\Vb}^{(1:l+1)}, \bar{\thetab} \big\}$.
\end{proof}

\begin{theorem}[Convergence to a critical point]
	\label{theorem:critical}
	The sequence $\big\{ \mathbf{h}_k^{(1:l)}, \Vb_k^{(1:l+1)}, \thetab_k \big\}_{k \ge 1}$ generated by Algorithm~\ref{alg:dnnnsr} converges to a critical point of the problem $\mathcal{P}_\mu$ defined in \eqref{eq:Q}.
\end{theorem}

\begin{proof}
	By Theorem~\ref{theorem:global}, the entire sequence converges to the limit point $\big\{ \bar{\mathbf{h}}^{(1:l)}, \bar{\Vb}^{(1:l+1)}, \bar{\thetab} \big\}$. By Theorem~\ref{theorem:subseq}(c), every limit point of the sequence is a critical point of $Q$. Therefore, the limit point $\big\{ \bar{\mathbf{h}}^{(1:l)}, \bar{\Vb}^{(1:l+1)}, \bar{\thetab} \big\}$ is a critical point of $\mathcal{P}_\mu$, which completes the proof. The detailed proof follows the same lines as \cite[Theorem~3]{amini2019}.
\end{proof}

\subsubsection*{Summary of the Convergence Analysis}

The convergence analysis presented above proceeds in four steps. First, Proposition~\ref{proposition_1} establishes that the partial derivatives of the smooth part $g$ are Lipschitz continuous, which justifies the choice of step sizes in the proximal updates. Second, Proposition~\ref{proposition_2} establishes the monotone descent and square-summability properties, which are the key technical ingredients for the subsequent analysis. Third, Proposition~\ref{proposition_3} and Theorem~\ref{theorem:subseq} establish boundedness, existence of limit points, convergence of regularizers, and criticality of limit points. Fourth, Proposition~\ref{proposition:kl} and Theorem~\ref{theorem:global} invoke the KL property and the general convergence theorem of \cite{whole_sequence} to upgrade the subsequential result to convergence of the entire sequence. The final result, Theorem~\ref{theorem:critical}, states that the DNN-NSR algorithm converges to a critical point of the penalty objective $\mathcal{P}_\mu$. This provides a rigorous theoretical foundation for the empirical results reported in Section V.

\section{Experimental Results}
In this section, a series of experiments on both synthetic and real data are conducted to evaluate the performance of the proposed matrix completion algorithm. We compare the results of the proposed algorithm with six algorithms, IALM \cite{IALM}, NCARL \cite{NCARL}, AEMC \cite{AEMC_DLMC}, DLMC \cite{AEMC_DLMC}, BiBNN \cite{bi_branch}, and LeRMC \cite{LeRMC}. The first two algorithms are of LRMC type, while the others are based on deep neural networks and NLMC. For simplicity and to ensure a balanced regularization effect across all layers, we empirically set $\alpha_i = \beta_j = 0.1$ for all $i$ and $j$ in our simulations, and $\lambda = 10^{-3}$ in \eqref{eq:reg_loss}.
The optimal value of the extrapolation weight, $\omega_{\theta, k}$, is different for various experiments and missing rates, so we set its optimal value for each experiment using search in $\big [\ 0, \ \omega_{\theta, k} \approxeq 0.5 \sqrt{\delta_{k} \frac{L_{\theta, k-1}}{L_{\theta, k}}} \big]$ interval with step-size $0.05$ , and the value with the highest convergence rate is chosen as the optimal value of $\omega_{\theta, k}$.
At the beginning of training, we set $\mu = \mu_{\max}$ to satisfy \eqref{eq:proof_eq54} (the descent lemma), and then gradually decrease it using a cosine annealing schedule to $\mu = \mu_{\min}$ in order to strengthen the effect of the regularization terms. We have chosen $\mu_{\min} = 1$ in all experiments, but $\mu_{\max}$ is sensitive to the missing rate and datasets and is determined accordingly. If the hyperparameter values differ from those mentioned earlier, it will be stated in the experiment.

\begin{table*}[tp]
	\small
	\caption{PSNRs, SDs, and MSEs on synthetic matrix $\mathbf{X}_{1} \in \mathbb{R}^{100 \times 200}$  with $\rho = 10\%,30\%,50\%$, and $80\%$.}
	\begin{center}
		\begin{tabular}{| c | c | c | c | c |c| c| c | c |}
			\hline
			$\rho (\%)$ & Metric & IALM & NCARL & AEMC & DLMC & BiBNN & LeRMC & \scriptsize DNN-NSR \\
			\hline
			\hline
			\multirow{3}{*}{$10$} & PSNR & $24.5461$ & $28.0245$ & $29.7751$ & $33.4047$ & $35.1094$ & $35.8012$ & $\mathbf{\textcolor{purple}{36.3523}}$ \\
			& SD & $0.8544$ & $0.7311$ &  $\mathbf{\textcolor{purple}{0.5079}}$ & $0.5659$ & $0.6021$ & $0.6242$ &$0.6071$ \\
			& MSE & $18.7516$ & $6.0520$ & $6.1220$ & $4.7950$ & $3.9784$ & $2.8456$ & $\mathbf{\textcolor{purple}{1.1040}}$ \\
			\hline
			\multirow{3}{*}{$30$} & PSNR & $19.1643$ & $22.9383$ & $23.7157$ & $27.8859$ & $31.4614$ & $32.1214$ & $\mathbf{\textcolor{purple}{33.2218}}$ \\
			& SD & $1.1022$ & $1.1880$ & $1.1918$ & $1.2249$ & $1.0837$ & $0.9435$ & $\mathbf{\textcolor{purple}{0.6140}}$ \\
			& MSE & $25.0865$ & $10.8373$ & $7.2170$ & $5.9620$ & $4.1327$ & $3.7168$ & $\mathbf{\textcolor{purple}{2.6990}}$\\
			\hline
			\multirow{3}{*}{$50$} & PSNR & $15.3393$ & $19.7097$ & $21.5625$ & $24.8806$ & $25.7349$ & $27.6214$ & $\mathbf{\textcolor{purple}{30.8950}}$ \\
			& SD & $0.7864$ &  $0.7116$ & $0.7566$ & $0.7640$ &  $0.6926$ & $0.6745$ & $\mathbf{\textcolor{purple}{0.6548}}$ \\
			& MSE & $35.1057$ & $16.8396$ & $11.4520$ & $8.7620$ & $8.4306$ & $7.8456$ & $\mathbf{\textcolor{purple}{5.6730}}$ \\
			\hline
			\multirow{3}{*}{$80$} & PSNR & $13.0411$ & $17.7830$ & $15.4138$ & $18.8448$ & $18.8266$ & $20.3245$ & $\mathbf{\textcolor{purple}{23.0441}}$ \\
			& SD & $0.8336$ & $0.8873$ & $1.1005$ & $1.0897$ & $1.2140$ & $1.0876$ & $\mathbf{\textcolor{purple}{0.8288}}$\\
			& MSE & $51.8842$ & $30.9560$ &  $28.4945$ & $22.1987$ & $23.8657$ & $21.3654$ & $\mathbf{\textcolor{purple}{19.5280}}$\\
			\hline    
		\end{tabular}
	\end{center}
	\label{table1_sim}
\end{table*}
\begin{table*}[tp]
	\small
	\caption{PSNRs, SDs, and MSEs on synthetic matrix $\mathbf{X}_{2} \in \mathbb{R}^{300 \times 200}$  with $\rho = 10\%,30\%,50\%$, and $80\%$.}
	\begin{center}
		\begin{tabular}{| c | c | c | c | c |c| c| c | c |}
			\hline
			$\rho (\%)$ & Metric & IALM & NCARL & AEMC & DLMC & BiBNN & LeRMC & \scriptsize DNN-NSR \\
			\hline
			\hline
			\multirow{3}{*}{$10$} & PSNR & $26.1469$ & $30.3217$ & $30.5333$ &  $32.5475$ & $35.2654$ & $36.2145$ & $\mathbf{\textcolor{purple}{37.3565}}$ \\
			& SD & $0.6883$ & $0.6617$ & $0.6655$ & $0.6592$ & $0.7553$ & $0.6784$ & $\mathbf{\textcolor{purple}{0.5672}}$\\
			& MSE & $15.4390$ & $4.8132$ & $4.5421$ & $3.7843$ & $3.2200$ & $2.5647$ & $\mathbf{\textcolor{purple}{1.9130}}$\\
			\hline
			
			\multirow{3}{*}{$30$} & PSNR & $20.5052$ & $24.4008$ & $24.0187$ & $27.4959$ & $32.9584$ & $33.5687$ & $\mathbf{\textcolor{purple}{35.1072}}$ \\
			& SD & $0.8169$ & $0.6915$ & $0.6056$ & $0.7185$ & $0.7552$ & $0.7225$ & $\mathbf{\textcolor{purple}{0.6594}}$\\
			& MSE & $18.2639$ & $6.8739$ & $6.4113$ & $4.3082$ & $3.4340$ & $3.0214$ & $\mathbf{\textcolor{purple}{2.7940}}$\\
			\hline
			
			\multirow{3}{*}{$50$} & PSNR & $17.8976$ & $21.5373$ & $18.1350$ & $25.7063$ & $27.1325$ & $28.1421$ & $\mathbf{\textcolor{purple}{29.8678}}$\\
			& SD & $0.9126$ & $0.9723$ & $0.6433$ & $0.6414$ & $\mathbf{\textcolor{purple}{0.6370}}$ &     $0.6516$ & $0.6749$\\
			& MSE &  $22.7209$ & $14.3484$ & $12.2660$ & $10.0870$ & $8.3160$ & $6.8475$ & $\mathbf{\textcolor{purple}{4.1823}}$ \\
			\hline
			
			\multirow{3}{*}{$80$} & PSNR &   $12.3161$ & $17.7434$ & $17.6322$ & $20.6940$ & $21.2342$ & $23.1774$ & $\mathbf{\textcolor{purple}{25.0967}}$\\
			& SD & $0.5247$ & $0.4662$ & $0.6706$ & $0.4890$ & $0.5081$ & $0.4420$ & $\mathbf{\textcolor{purple}{0.3929}}$\\
			& MSE & $48.8798$ & $27.1525$ & $27.9435$ & $24.9295$ & $24.4341$ & $23.0478$ & $\mathbf{\textcolor{purple}{18.5580}}$\\
			\hline
		\end{tabular}
	\end{center}
	\label{table2_sim}
\end{table*}

\subsection{Synthetic Matrix Recovery}

To evaluate the performance of the proposed method, we generate two synthetic matrices with an $m \times n$ matrix of rank $r$ as follows:
\begin{equation}\label{eq102}
\mathbf{X} = g \Bigg( 1.2 \Big( 0.5 g^{2} (\mathbf{A}\mathbf{B}) - g(\mathbf{A} \mathbf{B}) -1 \Big)  \Bigg) + \mathbf{A} \mathbf{B},
\end{equation}
where $\mathbf{A} \in \mathbb{R}^{m \times r}$ and $\mathbf{B} \in \mathbb{R}^{r \times n}$. Both matrices have independent and identically distributed (i.i.d.) standard Gaussian entries. $g(x) = 1.71 \tanh(\frac{2}{3} x)$ is an element-wise activation function as in \cite{AEMC_DLMC}. We use \eqref{eq102} to make $\mathbf{X}_{1} \in \mathbb{R}^{100 \times 200}$ and $\mathbf{X}_{2}\in \mathbb{R}^{300 \times 200}$ matrices, $r=10$. Then we randomly remove $10\%$, $30\%$, $50\%$, and $80\%$ of the entries.

The peak signal-to-noise ratio (PSNR) \cite{HTN} is employed to evaluate the performance of the proposed method for both synthetic matrix completion and image inpainting, defined as:
\begin{equation}\label{eq103}
\mathrm{PSNR} = 10 \log_{10} \frac{mn(\mathrm{max}({\mathbf{X}}))^{2}}{\lVert\widehat{\mathbf{X}} - \mathbf{X}\rVert_{F}^{2}},
\end{equation}
where $\mathbf{X}$ and $\widehat{\mathbf{X}}$ are the observed and the recovered matrices, respectively, and $\max({\mathbf{X}})$ is the maximum element in $\mathbf{X}$. The following MSE is used to evaluate the performance of matrix completion \cite{AEMC_DLMC}
\begin{equation}\label{eq104}
\mathrm{MSE} = \frac{\sum_{i,j \in \bar{\Omega}} (\widehat{x}_{ij} - x_{ij})^{2}}{\sum_{i,j \in \bar{\Omega}} x_{ij}^{2}}.
\end{equation}
Larger PSNR and smaller MSE reflect better recovery results. In this experiment, the PSNRs and MSEs presented in Tables \ref{table1_sim} and \ref{table2_sim} represent the mean values obtained from ten repeated trials each for LRMC and for training deep neural networks from scratch. We use a partially observed matrix to train the neural network. The stability of the proposed method is measured by the standard deviation (SD) of PSNRs.

In both synthesis matrices, $\mu_{\max} = 10^{5}$ for the missing rate $\rho = 80\%$, and $\mu_{\max} = 10^{6}$ for the rest of the missing rates.

The results reported in Tables \ref{table1_sim} and \ref{table2_sim} show that DNN-NSR consistently and significantly outperforms both the linear estimators (IALM and NCARL) and the deep nonlinear estimators (AEMC, DLMC, BiBNN, and LeRMC). We attribute this improvement to three complementary mechanisms that reinforce each other during training.

\textbf{Sparsity of hidden representations.} The $\ell_1$ penalty applied to the outputs of all hidden layers encourages the network to learn sparse latent features. Because the ground-truth matrices $\mathbf{X}_1$ and $\mathbf{X}_2$ are generated by the low-rank nonlinear model in \eqref{eq102}, their latent factors admit a sparse representation in the feature space induced by the activation function $g(x)=1.71\tanh(\tfrac{2}{3}x)$. By promoting sparsity at every layer, the $\ell_1$ penalty prevents the network from activating redundant neurons that would otherwise fit the noise and the particular realization of the observed entries. This is especially important for the synthetic matrices, whose intrinsic rank is only $r=10$; without such a penalty, a network with excessive effective capacity would readily memorize the observed entries when the sampling ratio is low.

\textbf{Low-rankness of weight matrices.} The nuclear norm penalty on the weight matrices imposes a soft low-rank constraint on each layer's linear transformation. Because the generative process in \eqref{eq102} is low-rank by construction, the nuclear norm serves as a convex surrogate of the rank function and aligns the nonlinear estimator with the intrinsic dimensionality of the data. This alignment is essential in the low-sampling regime: at $\rho=80\%$, only $20\%$ of the entries are observed, and an unregularized deep estimator has enough capacity to memorize them, which leads to a severe generalization error. The nuclear norm penalty prevents this by restricting the hypothesis space to approximately low-rank transformations, thereby steering the solution toward the true latent subspace.

\textbf{Gradual learning.} The most influential mechanism is the gradual learning strategy, in which the penalty parameter $\mu$ is annealed from $\mu_{\max}$ to $\mu_{\min}$. This schedule allows the network to first fit the dominant low-frequency structure of the observed entries and only then refine the solution as the influence of the nonsmooth regularizers increases. This curriculum-like behavior prevents the early epochs from being dominated by the regularizers, which would otherwise bias the network toward trivial or overly sparse solutions. In the later epochs, the regularization strength becomes sufficient to control the effective capacity of the model, resulting in a smooth transition from an essentially unregularized estimator to a strongly regularized one.

The empirical results strongly support this analysis. For $\mathbf{X}_1$, the PSNR gain of DNN-NSR over the best deep baseline (LeRMC) increases monotonically from $0.55$~dB at $\rho=10\%$ to $2.72$~dB at $\rho=80\%$. For $\mathbf{X}_2$, the gain rises from $1.14$~dB to $1.92$~dB over the same range. This monotonic growth with $\rho$ is a clear signature of over-fitting control: as the sampling ratio decreases, the over-fitting risk of unregularized deep estimators becomes more severe, and the relative benefit of the proposed regularizers therefore becomes more pronounced. In contrast, IALM and NCARL do not exhibit this behavior because they rely on a linear low-rank model that cannot capture the nonlinear structure induced by $g(\cdot)$ in \eqref{eq102}. Their performance is limited primarily by model mismatch rather than by over-fitting, which explains their consistently lower PSNR across all sampling ratios.

The MSE, which is computed exclusively over the unobserved entries, provides an independent and arguably more reliable indicator of over-fitting. As reported in Table \ref{table1_sim}, the MSE improvement ranges from \textcolor{black}{\(1.01\)} at $\rho=30\%$ to \textcolor{black}{\(1.74\)} at $\rho=10\%$, confirming that the proposed regularization reduces the generalization error rather than merely fitting the observed entries. Finally, the consistently lower standard deviation of PSNR across ten independent trials most notably at $\rho=80\%$ indicates that the regularization improves not only the mean recovery quality but also the stability of the training dynamics. This stability follows directly from the boundedness of the iterates induced by $r_\theta(\boldsymbol{\theta})$ and from the monotone descent property established in Proposition~2, which together suppress the erratic parameter updates typically observed in unregularized deep matrix completion models.

\subsection{Single-Image Inpainting}

\begin{figure}[h!]
	\centering
	\begin{subfigure}{3cm}
		\centering
		\includegraphics[height=2.5cm,width=2.5cm]{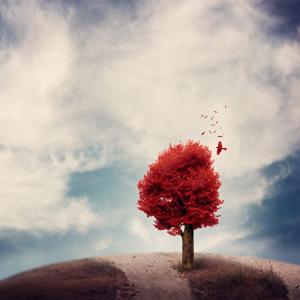}
		\caption{Image $\mathrm{I}$}
		\label{img1}
	\end{subfigure}
	\hspace*{0.25cm}
	\begin{subfigure}{3cm}
		\centering
		\includegraphics[height=2.5cm,width=2.5cm]{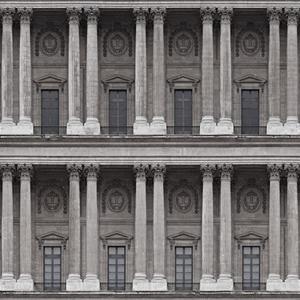}
		\caption{Image $\mathrm{II}$}
		\label{img2}
	\end{subfigure}
	\caption{RGB images for the experiment.}
	\label{images}
\end{figure}

\begin{figure}[h!]
	\centering
	\begin{subfigure}{3cm}
		\centering
		\includegraphics[height=2.5cm,width=2.5cm]{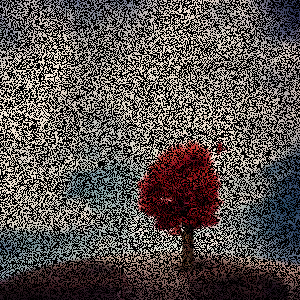}
		\caption{Image $\mathrm{I}$}
		\label{img3}
	\end{subfigure}
	\hspace*{0.25cm}
	\begin{subfigure}{3cm}
		\centering
		\includegraphics[height=2.5cm,width=2.5cm]{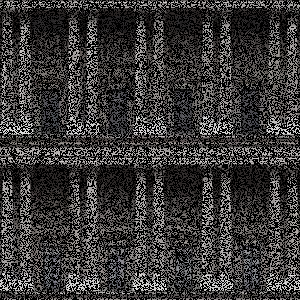}
		\caption{Image $\mathrm{II}$}
		\label{img4}
	\end{subfigure}
	\begin{subfigure}{3cm}
		\centering
		\includegraphics[height=2.5cm,width=2.5cm]{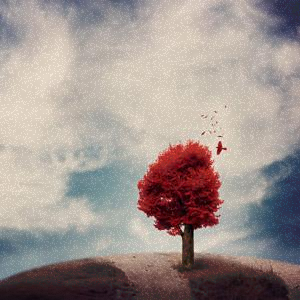}
		\caption{Image $\mathrm{III}$}
		\label{img5}
	\end{subfigure}
	\hspace*{0.25cm}
	\begin{subfigure}{3cm}
		\centering
		\includegraphics[height=2.5cm,width=2.5cm]{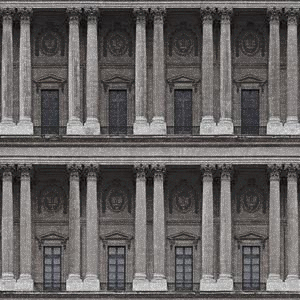}
		\caption{Image $\mathrm{IV}$}
		\label{img6}
	\end{subfigure}
	\caption{Masked and inpainted results for Images $\mathrm{I}$ and $\mathrm{II}$ with $50\%$ random masked pixels, (a): masked Image $\mathrm{I}$, (b) masked Image $\mathrm{II}$, (c) inpainted Image $\mathrm{I}$, (d) inpainted Image $\mathrm{II}$.}
	\label{masked_images}
\end{figure}

One of the applications of matrix completion problems is to recover the missing pixels of an image \cite{TNN}, \cite{HTN}. The pixel values of an RGB image are between 0 and 255, and each pixel is made up of three values red, green, and blue. Each channel of the image is a matrix with the dimensions of the image, and by concatenating the matrix of all three channels, a new matrix is obtained. We consider two images, $\mathrm{I}$ and $\mathrm{II}$, in Fig. \ref{images} and randomly mask their pixels with rates of $\rho = 30\%$, $40\%$, and $50\%$ and $\mu_{\max} = 10^{6}$.

Structure similarity (SSIM) \cite{HTN} is a measure used to evaluate the quality of the reconstructed image
\begin{equation}
\mathrm{SSIM} = \frac{(2 \mu_{\widehat{\mathbf{X}}} \mu_{\mathbf{X}} + C_{1})(2 \sigma_{\widehat{\mathbf{X}} \mathbf{X}} + C_{2})}{(\mu_{\widehat{\mathbf{X}}}^{2} + \mu_{\mathbf{X}}^{2} + C_{1}) (\sigma_{\widehat{\mathbf{X}}}^{2} + \sigma_{\mathbf{X}}^{2} + C_{2})},
\end{equation}
where $\mu_{\widehat{\mathbf{X}}}$ and $\mu_{\mathbf{X}}$ are mean values, $\sigma_{\mathbf{X}}^{2}$ and $\sigma_{\widehat{\mathbf{X}}}^{2}$ are variances, and $\sigma_{\widehat{\mathbf{X}} \mathbf{X}}$ is the covariance. $C_{1} = 0.01$ and $C_{2} = 0.03$ are two constant values that prevent the denominator from becoming zero. A higher value of SSIM represents a better recovery. The SSIM criterion is a number between $0$ and $1$.

We have used PSNR and SSIM criteria to compare DNN-NSR with six other algorithms. Reported values for both metrics are the average of 10 iterations of the algorithm run.

\begin{table*}[h!]
	\caption{PSNRs and SSIMs on Image $\mathrm{I}$ in Fig. \ref{images} with $30\%$, $40\%$, and $50\%$ randomly masked pixels as examples.}\label{table:im1}
	\begin{center}
		\begin{tabular}{|c|c|c|c|c|c|c|c|c|}
			\hline
			$\rho (\%)$ & Metric & IALM & NCARL & AEMC & DLMC & BiBNN & LeRMC & DNN-NSR \\
			\hline
			\hline
			\multirow{2}{*}{$30$} & PSNR & $16.7313$ & $27.6908$ & $28.4982$ & $29.6249$ & $30.4575$ & $30.5142$ & $\mathbf{\textcolor{purple}{30.8723}}$\\
			& SSIM & $0.3434$ & $0.7838$ & $0.8269$ & $0.8629$ & $0.8736$ & $0.8755$ & $\mathbf{\textcolor{purple}{0.8774}}$\\
			\hline
			
			\multirow{2}{*}{$40$} & PSNR & $16.7342$ & $26.6660$ & $27.9642$ & $29.3363$ & $29.7862$ & $29.9431$ & $\mathbf{\textcolor{purple}{30.2301}}$\\
			& SSIM & $0.3310$ & $0.7606$ &  $0.7946$ & $0.8488$ & $0.8597$ & $0.8610$ & $\mathbf{\textcolor{purple}{0.8656}}$\\
			\hline
			
			\multirow{2}{*}{$50$} & PSNR & $16.7315$ & $25.9480$ & $27.4329$ & $28.7921$ & $28.9233$ & $29.1411$ & $\mathbf{\textcolor{purple}{30.0301}}$\\
			& SSIM & $0.3169$ & $0.7466$ & $0.7221$ & $0.8009$ & $0.8342$ & $0.8411$ & $\mathbf{\textcolor{purple}{0.8521}}$\\
			\hline
			
		\end{tabular}
	\end{center}
	\label{table3_sim}
\end{table*}

\begin{table*}[h!]
	\caption{PSNRs and SSIMs on Image $\mathrm{II}$ in Fig. \ref{images} with $30\%$, $40\%$, and $50\%$ randomly masked pixels as examples.}\label{table:im2}
	\begin{center}
		\begin{tabular}{|c|c|c|c|c|c|c|c|c|}
			\hline
			$\rho (\%)$ & Metric & IALM & NCARL & AEMC & DLMC & BiBNN & LeRMC & DNN-NSR \\
			\hline
			\hline
			\multirow{2}{*}{$30$} & PSNR & $19.5454$ & $26.7432$ & $27.5896$ & $28.0446$ & $\mathbf{\textcolor{purple}{30.0519}}$ & $29.0135$ & $29.7962$\\
			& SSIM & $0.4912$ & $0.9177$ & $0.9364$ & $0.9501$ & $\mathbf{\textcolor{purple}{0.9797}}$ & $0.9587$ & $0.9693$\\
			\hline
			
			\multirow{2}{*}{$40$} & PSNR & $19.5414$ & $25.6884$ & $26.7736$ & $27.3264$ & $29.0604$ & $29.0847$ & $\mathbf{\textcolor{purple}{29.1035}}$\\
			& SSIM & $0.4904$ & $0.8891$ &  $0.9084$ & $0.9291$ & $0.9600$ & $0.9615$ & $\mathbf{\textcolor{purple}{0.9643}}$\\
			\hline
			
			\multirow{2}{*}{$50$} & PSNR & $19.5241$ & $25.4958$ & $27.0324$ & $26.6429$ & $27.9312$ & $28.1475$ & $\mathbf{\textcolor{purple}{28.8465}}$\\
			& SSIM & $0.4894$ & $0.8429$ & $0.8691$ & $0.9010$ & $0.9407$ & $0.9424$ & $\mathbf{\textcolor{purple}{0.9488}}$\\
			\hline
			
		\end{tabular}
	\end{center}
	\label{table4_sim}
\end{table*}

The results in Table \ref{table:im1} for Image $\mathrm{I}$ in Fig. \ref{images} show that the PSNR and SSIM gains of DNN-NSR over the best competing algorithm grow monotonically with the missing rate. Specifically, the PSNR improvement increases from $0.36$~dB at $\rho=30\%$ to $0.89$~dB at $\rho=50\%$, while the corresponding SSIM gains are $0.0019$, $0.0046$, and $0.011$ at $\rho=30\%$, $40\%$, and $50\%$, respectively. This monotonic trend is consistent with the over-fitting argument developed for the synthetic case: as more pixels are masked, the number of observed entries decreases, and the risk of memorizing the training data increases. The proposed regularizers counteract this risk by constraining the effective capacity of the network, and their relative benefit therefore becomes more pronounced at higher missing rates. The SSIM results are particularly informative because SSIM, unlike PSNR, is sensitive to the preservation of structural information such as edges, textures, and local contrast. The fact that DNN-NSR achieves the highest SSIM at $\rho=50\%$ indicates that the regularization preserves the perceptual structure of the image rather than merely minimizing pixel-wise error, which is a direct consequence of the sparsity prior imposed on the hidden representations.

The results in Table \ref{table:im2} for Image $\mathrm{II}$ in Fig. \ref{images} reveal a more nuanced picture. At $\rho=30\%$, BiBNN achieves the highest PSNR and SSIM, outperforming DNN-NSR by approximately $0.26$~dB in PSNR and $0.010$ in SSIM. This advantage can be attributed to the fact that BiBNN is specifically designed for matrix completion via a bilinear factorization architecture, which is well suited to the relatively dense sampling regime where the observed entries are sufficient to constrain the two branches. However, as $\rho$ increases to $40\%$ and $50\%$, DNN-NSR overtakes BiBNN and achieves the best performance among all algorithms. This crossover is a direct manifestation of over-fitting control: at low missing rates, BiBNN's inductive bias is sufficient, and the additional regularization of DNN-NSR provides little benefit; at high missing rates, by contrast, BiBNN's capacity becomes excessive relative to the limited observed data, and its performance degrades faster than that of DNN-NSR. The gradual learning schedule of DNN-NSR is particularly advantageous in this regime, because it allows the network to first capture the dominant low-frequency content of the image and only subsequently refine the high-frequency details under the increasing influence of the nonsmooth regularizers. This is consistent with the visual results in Fig. \ref{masked_images}, where the inpainted images at $\rho=50\%$ preserve both global structure and local texture.

It is also worth noting that IALM and NCARL perform poorly on both images, with PSNR values that are essentially independent of $\rho$ (e.g., $16.73$~dB for IALM at all three rates on Image $\mathrm{I}$). This behavior arises because RGB images are not exactly low-rank, so the linear low-rank assumption underlying these methods is severely violated. The resulting model mismatch, rather than over-fitting, is the dominant source of error for these algorithms, which explains why their performance does not improve as $\rho$ increases. In contrast, the nonlinear estimators (AEMC, DLMC, BiBNN, LeRMC, and DNN-NSR) all benefit from the ability to represent nonlinear image structures, but only DNN-NSR combines this representational power with explicit over-fitting control through the $\ell_1$ and nuclear norm regularizers and the gradual learning schedule.

Fig. \ref{masked_images} shows the inpainting result when $50\%$ of the pixels are masked, allowing a direct visual inspection of the reconstruction quality. In summary, the image-inpainting experiments on Images $\mathrm{I}$ and $\mathrm{II}$ provide further evidence for the effectiveness of the proposed regularization, while also revealing an interesting dependence on the statistical structure of the underlying image. Unlike the synthetic matrices, whose ground-truth rank is exactly known, natural RGB images are only approximately low-rank and exhibit strong local correlations across neighboring pixels. This distinction has important implications for the behavior of the competing algorithms.

\subsection{Recommender Systems}

Another application of matrix completion is recommender systems \cite{recommender_system_not_random}, \cite{collaborative_filtering_factorization}. Two famous datasets used in recommender systems are MovieLens $100$k and MovieLens $1$M. In these datasets, users assign scores between $1$ and $5$ to the movies on the Netflix website. Table \ref{table5_sim} shows some information about the datasets.

\begin{table*}[h!]
	\caption{Information of MovieLens $100$k and $1$M.}
	\begin{center}
		\begin{tabular}{|c|c|c|c|c|}
			\hline
			Dataset & $\#$ Users & $\#$ Films & $\#$ Rating & $\#$ Rating range\\
			\hline
			\hline
			MovieLens $100$k & 943 & 1682 & $10^{5}$ & 1-5 \\
			\hline
			MovieLens $1$M & 6040 & 3952 & $10^{6}$ & 1-5\\
			\hline
		\end{tabular}
	\end{center}
	\label{table5_sim}
\end{table*}

A recommender system can be modeled as an incomplete matrix, where the rows and columns of the matrix represent users and films. Thus we have a matrix whose number of available entries is tiny because the number of movies that each user views is small. To train the FCNN, we use $70\%$ and $50\%$ of the available entries for training and the remaining $30\%$ and $50\%$ for testing the DNN-NSR algorithm with $\mu_{\max} = 10^{5}$.

To measure the performance of the algorithms, we use the Normalized Mean Absolute Error (NMAE) criteria defined as \cite{AEMC_DLMC}:
\begin{equation}
\mathrm{NMAE} = \frac{1}{(x_{\mathrm{max}} - x_{\mathrm{min}}) \rvert\overline{\Omega}\lvert} \sum_{i,j \in \overline{\Omega}} \lvert\widehat{x}_{ij} - x_{ij}\rvert,
\end{equation}
where $x_{\max}$ and $x_{\min}$ are the upper and lower bounds of ratings, respectively, $\lvert \overline{\Omega} \rvert$ is the number of missing entries, and $\widehat{x}_{ij}$ and $x_{ij}$ are the actual and recovered ratings of movie $j$ by user $i$, respectively. The declared NMAE is the average of $10$ iterations of the algorithm, and a small NMAE represents better performance.

\begin{table*}[h!]
	\caption{NMAE(\%) of recommender system on MovieLens $100$k and MovieLens $1$M datasets.}
	\begin{center}
		\begin{tabular}{|c|c|c|c|c|c|c|c|c|}
			\hline
			Dataset & $\rho (\%)$ & IALM & NCARL & AEMC & DLMC & BiBNN & LeRMC & DNN-NSR \\
			\hline
			\hline
			\multirow{2}{*}{MovieLens $100$k} & $30\%$ & $19.69$ &  $19.01$ & $18.87$ & $18.39$ & $17.23$ & $16.85$ & $\mathbf{\textcolor{purple}{15.54}}$\\
			& $50 \%$ & $20.22$ & $19.61$ & $19.06$ & $18.75$ & $18.12$ & $17.74$ & $\mathbf{\textcolor{purple}{17.07}}$  \\
			\hline
			\multirow{2}{*}{MovieLens $1$M} & $30\%$ & $18.51$ & $18.25$ & $18.18$ & $18.03$ & $17.02$ & $16.94$ & $\mathbf{\textcolor{purple}{16.82}}$\\
			& $50 \%$ & $18.95$ & $18.78$ & $18.36$ & $18.31$ & $17.78$ & $17.56$ & $\mathbf{\textcolor{purple}{17.16}}$\\
			\hline
			
		\end{tabular}
	\end{center}
	\label{table6_sim}
\end{table*}

Table \ref{table6_sim} shows that the results of the DNN-NSR algorithm are superior. The reason is that the gradual nonsmooth regularization terms have controlled over-fitting. The recommender system experiments on MovieLens $100$k and MovieLens $1$M provide the most challenging testbed for the proposed algorithm, because the observed entries are extremely sparse the rating matrix has a density of less than $6\%$ for MovieLens $100$k and less than $4\%$ for MovieLens $1$M and the ratings are discrete values in $\{1,2,3,4,5\}$ rather than continuous measurements. In this regime, over-fitting is not merely a risk but an almost inevitable outcome for any high-capacity estimator, and the ability to control it is the decisive factor determining generalization performance.

The NMAE results in Table \ref{table6_sim} show that DNN-NSR achieves the lowest NMAE on both datasets and at both training ratios ($70\%$ and $50\%$). On MovieLens $100$k, the improvement over the best deep baseline (LeRMC) is $1.31$ percentage points at $\rho=30\%$ and $0.67$ percentage points at $\rho=50\%$. On MovieLens $1$M, the improvement is smaller in absolute terms ($0.12$ and $0.40$ percentage points, respectively) but still consistent across both training ratios. This dataset-dependent behavior can be explained by the difference in the number of available ratings: MovieLens $1$M contains ten times more ratings than MovieLens $100$k, so the effective sampling density is higher, and the over-fitting risk is correspondingly lower. Consequently, the relative benefit of the proposed regularization is more pronounced on the sparser MovieLens $100$k dataset, which is precisely what the over-fitting-control hypothesis predicts.

The comparison with the linear baselines is also instructive. IALM and NCARL perform substantially worse than all deep estimators, with NMAE values that are $2-4$ percentage points higher than those of DNN-NSR. This is because the rating matrix is not low-rank in the strict sense—user preferences are driven by latent factors that interact nonlinearly with item attributes—so the linear low-rank model is fundamentally misspecified. The nonlinear estimators (AEMC, DLMC, BiBNN, LeRMC, and DNN-NSR) all outperform the linear baselines, but among them, DNN-NSR is the only one that explicitly controls over-fitting through the combination of $\ell_1$ sparsity on hidden representations, nuclear norm regularization of weight matrices, and gradual learning. The fact that DNN-NSR achieves the best NMAE on both datasets and at both training ratios confirms that this combination is effective across different data regimes, from the relatively dense MovieLens $1$M to the extremely sparse MovieLens $100$k.

Finally, it is worth emphasizing that the NMAE metric is particularly sensitive to over-fitting because it is computed only over the held-out entries. A model that memorizes the training ratings will achieve a low training error but a high NMAE on the test set. The consistently lower NMAE of DNN-NSR therefore provides direct evidence that the proposed regularization reduces the generalization error rather than merely fitting the observed data. This is consistent with the theoretical analysis in Section IV, where the boundedness of the iterates and the monotone descent property were shown to prevent the parameter updates from diverging, thereby stabilizing the training process and improving the generalization performance.

\subsection{Effect of Gradual Regularization}

The central distinguishing feature of the proposed DNN-NSR algorithm is its gradual learning strategy, in which the penalty parameters $\mu_{h^{(i)}}$, $\mu_{v^{(j)}}$, and $\mu_{\theta}$ are annealed from a large initial value $\mu_{\max}$ to a small final value $\mu_{\min}$ over the course of training. This section provides a systematic empirical study of this mechanism by isolating the effect of $\mu_{\max}$ while holding $\mu_{\min}=1$ fixed. We vary the exponent of $\mu_{\max}$ from $10^{6}$ down to $1$ and report the resulting PSNR on both synthetic matrices and both RGB images. Each reported value is the average over ten independent runs of the network trained from scratch, which ensures that the observed trends are not artifacts of a particular initialization.

\begin{table*}[h!]
	\caption{PSNRs for two synthetic data with $\rho = 50\%$ and $80\%$ and Images $\mathrm{I}$ and $\mathrm{II}$ in Fig. \ref{images} with random masked pixels $40\%$ and $50\%$ for different values of $\mu_{\max}$ and $\mu_{\min} = 1$.}
	\begin{center}
		\begin{tabular}{|c|c|c|c|c|c|c|c|c|c|}
			\hline
			index & \diagbox{$\rho (\%)$}{$\mu_{\mathrm{max}}$} & $10^{6}$ & $10^{5}$ & $10^{4}$ & $10^{3}$ &  $10^{2}$ & $10$ & $1$ \\
			\hline
			\hline
			\multirow{2}{*}{$100 \times 200$} & $50$ &  $\mathbf{\textcolor{purple}{30.8950}}$ & $28.8245$ & $24.3096$ &  $20.8356$ &   $16.0930$ & $15.3492$ & $14.0109$\\
			& $80$ & $-$ & $\mathbf{\textcolor{purple}{23.0441}}$ & $20.4194$ & $17.6362$ & $15.8181$ &  $13.9772$ & $13.2007$\\
			\hline
			\multirow{2}{*}{$300 \times 200$} & $50$ &  $\mathbf{\textcolor{purple}{29.8678}}$ & $27.5959$ & $25.3325$ &  $20.5376$ & $19.4572$ & $17.4001$ & $16.9405$\\
			& $80$ & $-$ &  $\mathbf{\textcolor{purple}{25.0967}}$ & $22.1119$ & $18.9528$ &  $16.1477$ &   $15.5734$ & $14.1351$\\
			\hline
			\multirow{2}{*}{Image $\mathrm{I}$} & $40$ &  $\mathbf{\textcolor{purple}{30.0741}}$ & $27.7275$ & $26.2802$ &  $24.2606$ & $21.5966$ & $20.2467$ & $20.0874$\\
			& $50$ &  $\mathbf{\textcolor{purple}{30.0301}}$ & $27.0758$ & $25.8667$ & $24.1605$ & $21.1786$ &  $20.1080$ & $20.6373$\\
			\hline
			\multirow{2}{*}{Image $\mathrm{II}$} & $40$ & $\mathbf{\textcolor{purple}{29.1035}}$ & $27.3190$ & $26.9870$ & $24.0372$ & $20.9548$ & $20.0761$ & $20.2331$\\
			& $50$ &  $\mathbf{\textcolor{purple}{28.8465}}$ & $28.0330$ & $27.2082$ &  $24.0114$ & $21.4443$ & $21.2196$ & $20.7818$\\
			\hline
		\end{tabular}
	\end{center}
	\label{table7_sim}
\end{table*}

The results in Table \ref{table7_sim} provide strong and consistent evidence that the gradual learning strategy is a decisive factor in the performance of DNN-NSR. Across all eight experimental settings, the best PSNR is achieved at the largest admissible value of $\mu_{\max}$ that is, when the regularization is introduced most gradually and the performance degrades monotonically as $\mu_{\max}$ is reduced toward $\mu_{\min}$. This behavior can be understood through the interaction between the penalty schedule and the optimization landscape.

When $\mu_{\max}$ is large, the early epochs are dominated by the data-fidelity term $\mathcal{L}(\boldsymbol{\theta})$, and the network is free to learn the dominant low-frequency components of the latent structure without being prematurely constrained by the regularizers. As training progresses, the penalty parameters decrease according to the cosine annealing schedule, and the effect of the $\ell_1$ and nuclear norm regularizers is gradually introduced. This curriculum-like schedule allows the network to first escape the region of the parameter space where the loss is dominated by the training residual and only subsequently refine the solution under the increasing influence of the nonsmooth regularizers. In effect, the network is allowed to explore a large hypothesis space in the early epochs and is then progressively confined to a region of low effective capacity, which is precisely the behavior needed to balance approximation and estimation error in the low-sampling regime.

In contrast, when $\mu_{\max}$ is small, the regularizers are active from the very first epoch. This premature regularization biases the network toward a restricted hypothesis space before it has had the opportunity to fit the dominant structure of the observed entries. The resulting solution is overly smooth and fails to capture the finer details of the latent structure, which explains the monotonic degradation of PSNR as $\mu_{\max}$ decreases from $10^{6}$ to $1$. The fact that the best performance is always obtained at the largest tested value of $\mu_{\max}$ even though the difference between consecutive exponents is an order of magnitude indicates that the network is relatively insensitive to the exact value of $\mu_{\max}$ as long as it is sufficiently large, but highly sensitive to the regime in which the regularizers are introduced too early.

The synthetic matrix experiments reveal an additional and more subtle effect. For the $100\times 200$ matrix at $\rho=50\%$, the best PSNR is obtained at $\mu_{\max}=10^{6}$, whereas at $\rho=80\%$ the best PSNR is obtained at $\mu_{\max}=10^{5}$. The same pattern is observed for the $300\times 200$ matrix: the optimal $\mu_{\max}$ shifts from $10^{6}$ at $\rho=50\%$ to $10^{5}$ at $\rho=80\%$. This shift can be explained by the sampling-ratio dependence of the annealing rate. When $\rho$ is large, the number of observed entries is small, and the risk of over-fitting is correspondingly high. If $\mu_{\max}$ is chosen too large (e.g., $10^{6}$) in this regime, the penalty parameters decrease slowly, and the regularizers do not become strong enough within the allotted training budget to prevent the network from memorizing the observed entries. As a result, the network enters the late training epochs with a capacity that is still too high relative to the amount of observed data, which leads to over-fitting and a degradation of the test PSNR. Choosing a smaller $\mu_{\max}$ (e.g., $10^{5}$) accelerates the annealing schedule, so that the regularizers become effective earlier in training, which compensates for the reduced amount of observed data. This is a direct manifestation of the trade-off between exploration and regularization: the optimal schedule must be adapted to the sampling ratio, and the fact that the optimal $\mu_{\max}$ decreases as $\rho$ increases is precisely what the over-fitting-control hypothesis predicts.

The image-inpainting results exhibit a qualitatively similar but quantitatively weaker dependence on $\rho$. For Image $\mathrm{I}$, the best PSNR is obtained at $\mu_{\max}=10^{6}$ for both $\rho=40\%$ and $\rho=50\%$, with the difference between $10^{6}$ and $10^{5}$ being approximately $2.3$ dB and $3.0$ dB, respectively. For Image $\mathrm{II}$, the best PSNR is also obtained at $\mu_{\max}=10^{6}$ for both missing rates, with the difference between $10^{6}$ and $10^{5}$ being approximately $1.8$ dB and $0.8$ dB. The fact that the optimal $\mu_{\max}$ does not shift to smaller values as $\rho$ increases in the image case can be attributed to the strong local correlations present in natural images, which effectively increase the amount of information available to the network beyond what the raw sampling ratio suggests. In other words, the intrinsic dimensionality of the image-inpainting problem is lower than that of the synthetic problem, so the network can tolerate a slower annealing schedule without over-fitting. This is consistent with the results in Tables \ref{table:im1} and \ref{table:im2}, where DNN-NSR achieves the best performance at all missing rates for Image $\mathrm{I}$ and at $\rho=40\%$ and $\rho=50\%$ for Image $\mathrm{II}$.

It is also worth noting that the degradation of PSNR at small $\mu_{\max}$ is considerably less severe for the images than for the synthetic matrices. For example, on the $100\times 200$ synthetic matrix at $\rho=50\%$, the PSNR drops from $30.90$ dB at $\mu_{\max}=10^{6}$ to $14.01$ dB at $\mu_{\max}=1$, a reduction of nearly $17$ dB. On Image $\mathrm{I}$ at $\rho=50\%$, the corresponding drop is from $30.03$ dB to $20.64$ dB, a reduction of approximately $9.4$ dB. This difference reflects the fact that the synthetic problem is more sensitive to premature regularization because its latent structure is purely low-rank and the network has no side information to compensate for the loss of effective capacity. RGB images, by contrast, contain redundant local structure that partially compensates for the reduced capacity, so the performance penalty of early regularization is less severe.

Taken together, these observations confirm that the gradual addition of the nonsmooth regularization terms is not merely a heuristic but a principled mechanism for controlling the effective capacity of the network during training. The optimal schedule depends on both the intrinsic dimensionality of the problem and the sampling ratio, and the ability to tune $\mu_{\max}$ accordingly is what allows DNN-NSR to achieve consistent gains over the competing algorithms across all three application domains studied in this paper.

\subsection{Effect of the Extrapolation Weight}

The extrapolation weight $\omega_{\theta, k}$ controls the momentum-like term in \eqref{eq:extrapolated_theta}, which combines the current and previous parameter estimates to form the extrapolated point $\widehat{\boldsymbol{\theta}}_{k-1}$. Since this term directly affects the effective step size of the proximal update in \eqref{eq:theta_extrapolated}, it plays a critical role in determining the convergence rate of DNN-NSR. In this experiment, we examine the sensitivity of the algorithm to $\omega_{\theta, k}$ on three representative settings:
\begin{enumerate}
	\item the sparse synthetic case: the $300 \times 200$ synthetic matrix with $\rho = 80\%$ and $\mu_{\max} = 10^{5}$;
	\item the dense synthetic case: the $300 \times 200$ synthetic matrix with $\rho = 30\%$ and $\mu_{\max} = 10^{6}$;
	\item the image-inpainting case: Image $\mathrm{I}$ in Fig. \ref{images} with $50\%$ random masked pixels and $\mu_{\max} = 10^{6}$.
\end{enumerate}

\begin{figure*}[h!]
	\centering
	\begin{subfigure}[b]{0.3\textwidth}
		\centering
		\includegraphics[width=\textwidth]{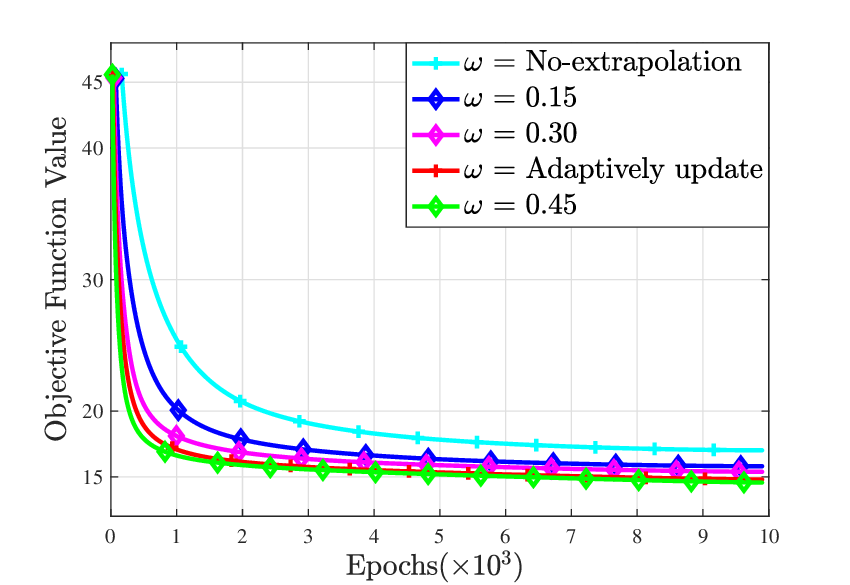}
		\caption{$300 \times 200$ synthesis matrix}
		\label{extrapolation1}
	\end{subfigure}
	\hspace*{0.25cm}
	\begin{subfigure}[b]{0.3\textwidth}
		\centering
		\includegraphics[width=\textwidth]{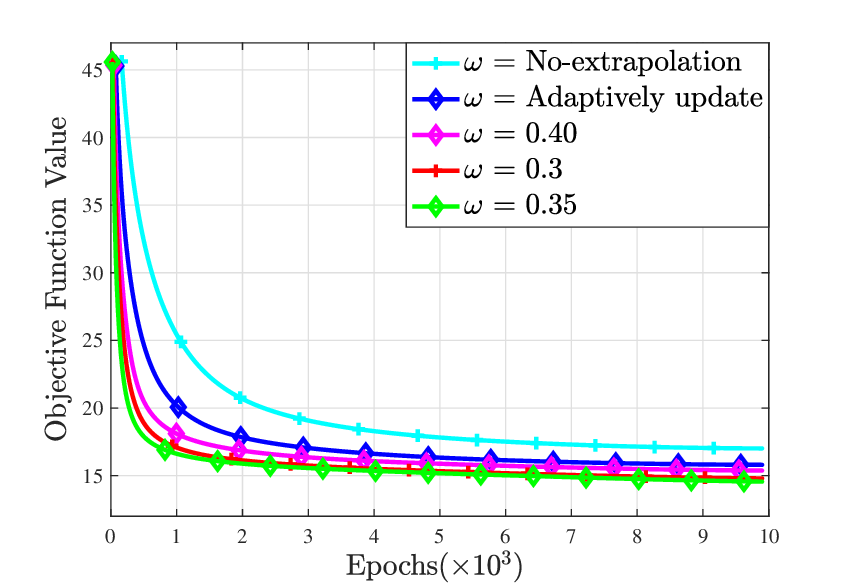}
		\caption{$300 \times 200$ synthesis matrix}
		\label{extrapolation2}
	\end{subfigure}
	\begin{subfigure}[b]{0.3\textwidth}
		\centering
		\includegraphics[width=\textwidth]{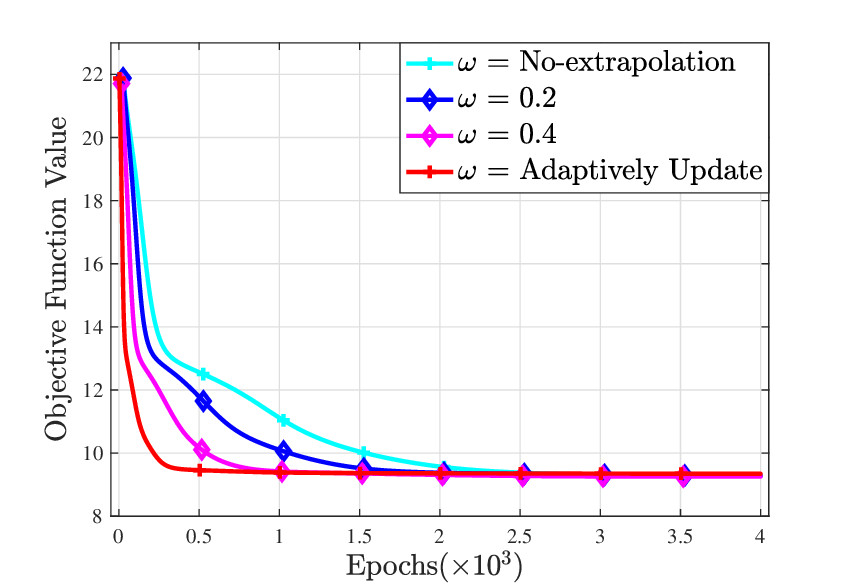}
		\caption{Image inpainting for Image $\mathrm{I}$}
		\label{extrapolation3}
	\end{subfigure}
	\caption{Variation of loss function for (a) $300 \times 200$ synthetic matrix with $\rho = 80\%$; (b) $300 \times 200$ synthetic matrix with $\rho = 30\%$, and (c) Image $\mathrm{I}$ in Fig. \ref{images} with $\rho = 50\%$.}
	\label{extrapolation weight}
\end{figure*}

Fig. \ref{extrapolation weight} shows the evolution of the loss function versus training epoch for different values of $\omega_{\theta, k}$, including the case $\omega_{\theta, k}=0$ (no extrapolation). Three observations can be made. First, in all three settings—the sparse synthetic case, the dense synthetic case, and the image-inpainting case—any nonzero extrapolation weight leads to a faster reduction of the loss in the early epochs. This confirms that the extrapolated point $\widehat{\boldsymbol{\theta}}_{k-1}$ carries useful information from the previous iteration that the plain proximal update cannot exploit. This observation is consistent with the established acceleration properties of extrapolation in proximal gradient methods. Second, the optimal extrapolation weight is problem-dependent: the best values for the sparse synthetic case, the dense synthetic case, and the image-inpainting case are $0.45$, $0.35$, and $0.5\sqrt{\delta_{k} L_{\theta, k-1}/L_{\theta, k}}$, respectively. The smaller optimal weight for the dense synthetic case ($\rho=30\%$, $\mu_{\max}=10^{6}$) reflects the fact that, in this relatively dense and strongly regularized regime, the parameter updates are already well-conditioned, and an overly aggressive extrapolation can overshoot the proximal step and slow down convergence. Third, and most importantly, the adaptive schedule $\omega_{\theta, k} = 0.5\sqrt{\delta_{k} L_{\theta, k-1}/L_{\theta, k}}$, which is derived from the theoretical upper bound in Proposition 2, achieves performance comparable to the best fixed value in every case. This is significant because it eliminates the need for a problem-specific hyperparameter search: the adaptive rule automatically adjusts the extrapolation strength to the local Lipschitz constant of the loss, thereby preserving the descent guarantee of Proposition 2 while retaining the acceleration benefit of extrapolation.

In summary, the extrapolation weight provides a simple yet effective mechanism for accelerating the convergence of DNN-NSR, and the adaptive rule derived from the convergence analysis offers a principled and dataset-independent choice that performs near-optimally across all tested scenarios.

\section{Conclusion}

In this paper, we proposed DNN-NSR, a deep FCNN with nonsmooth regularization for NLMC. To mitigate over-fitting caused by training on incomplete observations, we incorporated two complementary regularization terms into the network: the $\ell_1$-norm of intermediate representations to promote sparse latent features and the nuclear norm of layer-wise weight matrices to encourage low-rank transformations. By constraining both the learned representations and the transformations that generate them, the proposed formulation provides a structure-aware approach to controlling the complexity of deep NLMC estimators.

To optimize the resulting nonsmooth and nonconvex objective, we developed an extrapolated proximal gradient algorithm that employs soft-thresholding and singular-value thresholding for the two regularization terms, respectively. We also introduced an adaptive rule for the extrapolation weight based on the descent conditions of the algorithm. In addition, a gradual regularization strategy was incorporated to progressively increase the influence of the structural penalties during training, balancing data fitting and regularization.

The proposed method was evaluated on synthetic matrix recovery, single-image inpainting, and recommender system tasks using the MovieLens 100K and MovieLens 1M datasets. The results demonstrate that jointly exploiting sparse intermediate representations and low-rank weight structures can improve recovery performance compared with representative linear and nonlinear matrix completion methods. Although the performance advantage varies across datasets and observation rates, the experiments support the effectiveness of the proposed structure-aware regularization framework. They also illustrate the importance of gradually incorporating regularization during training.

Several directions remain for future research. One possibility is to replace the nuclear norm with nonconvex rank surrogates, such as the Schatten-$p$ quasi-norm for $0<p<1$, to obtain a closer approximation to the rank function. Another direction is to develop an adaptive regularization schedule that adjusts the penalty strengths during training based on validation performance. Finally, extending the proposed framework to structured matrix completion problems, including Hankel and Toeplitz matrix recovery in signal processing, is a promising direction for further investigation.

\section*{ACKNOWLEDGMENT}
The fourth author was partially supported by the Research Foundation Flanders (FWO) grant
grant G$081222$N and by the BOF DocPRO$4$ projects with ID $46929$ and $48996$.

\bibliographystyle{IEEEtran}
\bibliography{myref_MC}

\vfill\pagebreak

\end{document}